\documentclass[journal,comsoc]{IEEEtran}
\usepackage{url}
\usepackage{amsmath}
\usepackage{amssymb}
\usepackage{amsthm}
\usepackage{float}
\usepackage{graphics}
\usepackage{graphicx}
\usepackage{times}
\usepackage{algorithm,algorithmic}

\newtheorem{theorem}{Theorem}

\newtheorem{Definition}[theorem]{Definition}
\newtheorem{Lemma}[theorem]{Lemma}
\newtheorem{Remark}[theorem]{Remark}

\begin{document}

\title{On the Energy Efficient Displacement of Random Sensors for Interference and Connectivity} 
\author{Rafa\l{} Kapelko
\IEEEcompsocitemizethanks{\IEEEcompsocthanksitem 
R. Kapelko is with the Department of Computer Science,
Faculty of Fundamental Problems of Technology, Wroc{\l}aw University of Science and Technology, Poland.\protect\\
E-mail: rafal.kapelko@pwr.edu.pl.}
}
\maketitle

\begin{abstract}
This paper investigates the problem of the minimilization of energy consumption in reallocation of wireless mobile sensors network (WMSN) to assure good communication without interference.

Fix $d\in\mathbb{N}\setminus\{0\}.$ Assume  $n$ sensors are initially randomly placed in the hyperoctant $[0,\infty)^d$ according to $d$ identical and independent Poisson processes 
each with arrival rate $\lambda>0.$ 

Let $0< s \le v$ be given real numbers. We are allowed to move the sensors, so that 
every two consecutive sensors are placed at distance greater than or equal to $s$ and less than or equal to $v.$

Fix $a\ge 1.$ Assume that $i-$th  sensor is displaced a distance equal to $m(i).$ The cost measure for the  displacement of the team of sensors is the sum $\sum_{i=1}^{n}d_i^a$ ($a-$total movement).

In this work, we discover and explain \textbf{a sharp decline} and \textbf{a sharp increase} (a threshold phenomena) in the expected minimal $a-$total movement around the interference-connectivity distances $s,v$ equal to $\frac{1}{\lambda}.$
\end{abstract}


\IEEEpeerreviewmaketitle

\section{Introduction}\label{sec:introduction}
Wireless mobile sensors network (WMSN) (e.g. see \cite{abbasi2009movement, trans2018,  kapelko_sensor2018, kim2017, li2016} and \cite{mohamed2017}) are being deployed for detecting and monitoring events which occur in many instances of every day life. 
However, it is often their case that monitoring may not be as effective due to 
external factors such as harsh environmental conditions, sensor faults, geographic obstacles, etc.
In such cases sensor realignments may be required, e.g., sensors must be relocated from their initial positions
to new positions so as to attain the desired communication characteristics.
The resulting problem is assigning final positions to the sensors in order to minimize the reallocation cost.

Fix $d\in\mathbb{N}\setminus\{0\}.$
The present paper is concerned with random realignments of sensors in the hyperoctant. Assume that $n$ sensors are initially placed
in the hyperoctant $[0, \infty)^d$ 
according to $d$ identical and independent Poisson processes each with arrival rate $\lambda>0.$
\footnote{\textbf{Lets take the airplane which randomly droppes mobile sensors. This is the case $d=1$ of our model.}}

Given that the sensors are initially placed on the domain at random according to some well-defined random process, we are interested
to ensure that by moving the sensors the following scheduling requirement is satisfied.

\begin{Definition}[$(s,v)-IP$]
The $(s,v)-$\textit{interference-connectivity problem} 
requires that every two consecutive sensors\footnote{{The precise meaning of \textit{consecutive sensors} in the higher dimension\\ ($d>1$) will be explain further
in Definition \ref{def:atotek}  from Section \ref{chap:higher}.} } are placed at distance greater than or equal to $s$
and less than or equal to $v$
for some 
$s,v$ such that
$0<s\le v.$
\end{Definition}

It is very well known that proximity between sensors affects transmission and reception signals and causes
degradation of performance (see \cite{gupta}). The closer the distance between neighbouirng sensors, the higher the resulting
interference. 

Additionally, a typical sensor is able to sense and monitor a bounded region \cite{Hu:2014}, \cite{TIAN}.
In the theoretical model, the sensing area of each sensor node is a disk of radius $r$ (see \cite{kumar2005}).
Therefore, to ensure a good communication or connection of the whole network
the sensors can not be too far from each other. 

Let us consider the following simple example. Fix $k\in\mathbb{N}\setminus\{0\}.$ Assume that the sensors on the $[0,\infty)$ fullfil $(s,v)-IP$ requirement.
In this ideal case the sensor radius equal to $\frac{kv}{2}$ is enough for $k-$connectivity of the whole network, i.e. every point of the monitoring region is within the radius of at least $k$ sensors.
It is not difficult to see that, the similar argument also holds in the higher dimensions.

Thus, in $(s,v)-IP$ problem the goal is to ensure a good communication of the whole network while at the same time the consecutive sensors are not too close.

The initial placement of the sensors does not guarantee $(s,v)-IP$ interference-connectivity requirement
since the sensors have been placed \textbf{randomly} according to the arrival times of Poisson processes. 

Clearly, a sufficient density of random nodes is necessary to achieve property that the distance between two sensors is less than $v$
(see \cite{wang2009})
and some authors have proposed using several rounds of random displacement for desired connectivity \cite{eftekhari13,yan}.
Another approach is to have the sensors move from their initial location to a new position so as to achieve the desired connectivity property.
(e.g. see \cite{MinMax,barrierMinSum,eftekhari13d,Liu2005} and \cite{Saipulla2010}).

Obviously, the several rounds of random dispersal cause proximity between sensors. Hence to achieve property that any pair of sensors are separated by a distance of at least $s$
the sensors have to relocate.

Let us recall that, in this study  $n$ sensors are initially placed
in the hyperoctant $[0, \infty)^d$ 
according to $d$ identical and independent Poisson processes each with arrival rate $\lambda>0.$

To attain the requirements that any pair of sensors are separated by a distance of at least $s$ and no two consecutive sensors are placed
at distance greater than
$v,$ the sensors have to move from their initial random location to new position.

The fundamental problem is {\em the energy consumption} of the displacement of the team of sensors.
We define the cost measure 
{\em $a-$total movement} as follows.
\begin{Definition}[$a-$total movement]
\label{def:atot}
Let $a\ge 1$ be a constant. Suppose that the $i-$th sensor's displacement is equal to $|m(i)|.$
The {\em $a-$total movement} is defined as the sum $M_a:= \sum_{i=1}^n |m(i)|^a$. 
\end{Definition}
The main question we address in this paper is the following. Assume that $n$ sensors are initially placed
in the hyperoctant $[0, \infty)^d$ 
according to $d$ identical and independent Poisson processes each with arrival rate $\lambda>0.$
What is the expected minimal $a-$total movement 
so as to solve the problem $(s,v)-IP$
as a function of the parameter $s, v, n, \lambda, a, d$?
Our goal is to investigate tradeoffs arising among the parameters $s, v, n, \lambda, a, d.$

Let us fix $\epsilon, \tau >0$ arbitrary small constants independent on the number of sensors.
We explain the \textit{threshold phenomena} around the interference-connectivity distances $s=v=\frac{1}{\lambda}$
for the expected minimal $a-$total movement.
\subsection{Related Work}
Interference has been the subject of extensive interest in research community in the last decade. 
Some papers study interference in relation to network performance degradation \cite{gupta, degra08}.
Moscibroda et al. \cite{mosci02} consider the average interference problem while maintaining the desired network properties such as connectivity,
multicast trees or point-to-point connections, while in \cite{burkh02} the authors propose connectivity preserving and spanner constructions which are interference
optimal. The interference minimization in wireless ad-hoc networks in a plane was studied in \cite{halld07}. 
Further, \cite{Fu_2014, acm_2016} investigates  the  problem  of  optimally
determining source-destination connectivity in random networks.

More importantly, our work is related to the paper \cite{kranakis_shaikhet} where the authors consider the expected minimal total displacement
required so that every pair of sensors are in their final positions at distance greater or equal to $s$ for $n$ sensors placed uniformly according to Poisson process with arrival rate $\lambda=n.$

Compared to the minimum distance between two sensors, the $(s,v)-IP$ scheduling requirement not only avoids interference, but also ensures good connectivity and is more reasonable when considering interference.
Our analysis also generalizes the result of the paper \cite{kranakis_shaikhet} from $a=1$ to all exponents $a\ge 1$ and all Poisson processes with arrival rate $\lambda>0$  for more realistic
expected minimal $a-$total movement.
It is worth mentioning that some asymptotic bounds in \cite{kranakis_shaikhet} are one-sided.
We give full asymptotic results (lower and upper bound, exact asymptotics) which explain the 
tradeoffs arising among the parameters $s, v, n, a, d.$

Connectivity and barrier coverage have been extensively studied in research community (e.g., see \cite{barriercoverageNodeDegree,swat2012,cortes2012,Dobrev2017,sajal2008,Johnson2012,kumar2005} and \cite{saipulla2009}).
The work by \cite{kumar2005} introduced two notions of probabilistic barrier coverage: weak and strong barrier coverage.
In \cite{Lazos2006} Lazos et al. derived analytical expressions of coverage for heterogeneous sensor networks
The paper \cite{siamcontrol2015} investigated the problem of placing unreliable sensors in the unit interval to optimize the maximum cost.

Finally, it is worth mentioning that, our work is closely related  to the series of papers \cite{ICDCNkapelko}, \cite{pervasiveKAPELKO},  \cite{KK_2016_cube}, \cite{kapelkokranakisIPL}.
In \cite{kapelkokranakisIPL} the problem of energy consumption
of random sensors 
is analyzed to cover a unit interval and in  \cite{KK_2016_cube}
to provide full coverage of $d-$dimensional cube. The paper \cite{ICDCNkapelko} investigated the maximum of the expected sensor's displacement (the time required) for coverage with interference on the line.
Further, \cite{pervasiveKAPELKO} considered the maximum of the expected sensor's displacement to the power.

\subsection{Outline and Results of the Paper}
\label{chap:inter}
Throughout this paper $\epsilon, \tau>0$ are arbitrary small constants independent on $n.$

Fix $d\in\mathbb{N}\setminus\{0\}.$ Let $a\ge 1$ be a constant. Assume  $n$ sensors are initially placed in the hyperoctant $[0,\infty)^d$ according to $d$ identical and independent Poisson processes 
each with arrival rate $\lambda.$ We want to have the sensors moving from their current random locations to positions to ensure $(s,v)-IP$ interference-connectivity requirement.

We derive tradeoffs between the expected minimal $a-$total movement
and the interference-connectivity distances $s,v.$  Table \ref{tab:higher} 
summarizes the results proved in Section \ref{chap:line} and Section \ref{chap:higher}.\footnote{We recall the following asymptotic notation:
$(i)$ $f(n)=O(g(n))$ if there exists a constant $C_1>0$ and integer $N$ such that $|f(n)|\le C_1|g(n)|$ for all $n>N,$
$(ii)$ $f(n)=\Omega(g(n))$ if there exists a constant $C_2>0$ and integer $N$ such that $|f(n)|\ge C_2|g(n)|$ for all $n>N,$
$(iii)$ $f(n)=\Theta(g(n))$ if and only if $f(n)=O(g(n))$ and $f(n)=\Omega(g(n)).$}
\begin {table}[H]
\caption {The expected minimal $a-$total movement ($a\ge 1$)  of $n$ sensors in the hyperoctant $[0,\infty)^d$ as a function of the interference-connectivity distances $s, v.$ } \label{tab:higher} 
\begin{center}
\label{tab:dwa}
 \begin{tabular}{|c|c|c|} 

 \hline
 \begin{tabular}{c}Interference-\\-connectivity\\ distances $s, v$\end{tabular}  & \begin{tabular}{c}Expected minimal\\  $a-$total movement\end{tabular} &  Algorithms \\ 
  \hline
 \begin{tabular}{c}$s=\frac{1-\epsilon}{\lambda},$ $v=\frac{1+\tau}{\lambda}$\\ $\epsilon,\tau>0$\end{tabular} & $O(n)/\lambda^a$ & $I_d\left(n,\frac{1-\epsilon}{\lambda},\frac{1+\tau}{\lambda}\right)$  \\ 
 \hline
 \begin{tabular}{c}$s=\frac{1}{\lambda},$ $v=\frac{1}{\lambda}$\end{tabular} & 

$\Theta\left(n^{1+\frac{a}{2d}}\right)/\lambda^a$  

 & $MV_d\left(n,\frac{1}{\lambda}\right)$ \\  
 \hline
 \begin{tabular}{c}$s=\frac{1+\epsilon}{\lambda},$ $v=\frac{1+\tau}{\lambda}$  \\ $\tau\ge \epsilon>0$\end{tabular} & 
$\Theta\left(n^{1+\frac{a}{d}}\right)/\lambda^a,$   & 
$MV_d\left(n,\frac{1+\epsilon}{\lambda}\right)$ \\ 
 \hline
 \end{tabular}
\end{center}
\end{table}
Let us consider case $a=2$ and $d=1.$ We prove the following results.
\begin{itemize}
\item For interference-connectivity distances $s=v=\frac{1}{\lambda}$ the expected minimal $2-$total movement is in $\Theta\left(n^2\right)/{\lambda^2}.$
\item If $s=\frac{1-\epsilon}{\lambda}$ is below $\frac{1}{\lambda}$ and $v=\frac{1+\tau}{\lambda}$ is above $\frac{1}{\lambda},$ the expected minimal $2-$total movement declines sharply to $O\left(n\right)/\lambda^2.$
\item If both interference-connectivity distances $s=\frac{1+\epsilon}{\lambda}$ and $v=\frac{1+\tau}{\lambda}$ are above $\frac{1}{\lambda},$ the expected minimal $2-$total movement increses sharply to 
$\Theta\left(n^3\right)/\lambda^2.$ 
\end{itemize}
Similar sharp decrease and increase hold in all dimensions ($d\in \mathbb{N}\setminus \{0,\}$) and for all exponents $a\ge 1.$ 
Hence,  our investigations explain the \textit{threshold phenomena} around the interference-connectivity distances equal to $\frac{1}{\lambda}$ 
as this affects the expected minimal $a-$total movement of the sensors to fullfil $(s,v)-IP$ interference-connectivity requirement on the line and in the higher dimension.

Here is an outline of the paper. In Section \ref{chap:first} we provide several preliminary facts that will be used in the sequel.
In Section \ref{chap:line} we investigate sensors on the line. 
In Section \ref{chap:higher} we investigate sensors in the higher dimensions.
Section \ref{sec:sim} deals with the simulation results of our Algorithms (\ref{alg_one}-\ref{alg_interference}).
\section{Model and preliminaries.}
\label{chap:first}
In this section we recall some useful properties of the Poisson process and only some basic facts about special numbers and random variables which will be useful in the analysis in the next sections.

We consider $n$ random sensors 
initially placed in the half-infinite interval $[0,\infty)$ according to
Poisson process with arrival rate $\lambda>0.$
Assume that, the $i-$th event represents the location of the $i-$th sensor, for
$i=1,2,\dots,n.$ 

Let $X_i$ be the arrival time of the $i-$th event in this Poisson process, i.e., the position of the $i-$th sensor in the interval $[0,\infty).$
We know that the random variable $X_i$ obeys the Gamma distribution with parameters $i, \lambda.$
Its probability density function is given by $f_{i,\lambda}(t)=\lambda e^{-\lambda t}\frac{(\lambda t)^{i-1}}{(i-1)!}$
and $\Pr\left[X_i\ge t\right]=\int_t^{\infty}\lambda e^{-\lambda t}\frac{(\lambda t)^{i-1}}{(i-1)!}.$
Moreover, the following identity holds 
\begin{equation}
\label{eq:probadens} 
X_{j+l}-X_j=X_l,
\end{equation}
provided that $j,l\in\mathbf{N_+},$ 
(see \cite{kapelkogamma, kingman,dam_2014, ross_2002} for additional details on the Poisson process). Notice that,
\begin{equation}
\label{integral_2}
\int_0^{\infty}t^bf_{l,\lambda}(t)dt=\frac{1}{\lambda^b}\frac{(l-1+b)!}{(l-1)!},
\end{equation}
where $b$ is non-negative integer and $l,n$ are positive integers ( see \cite[Chapter 15]{stat_2011}).


We will use the following notations
for the rising factorial \cite{concrete_1994}
$$n^{\overline{k}} = \begin{cases} 1 &\mbox{for } k=0 \\
n(n+1)\dots(n+k-1) & \mbox{for } k\ge 1. \end{cases}$$
Let ${ n\brack k}$  be the Stirling numbers of the first kind, which are
defined for all integer numbers such that $0\le k \le n.$ 

The Stirling numbers of the first kind arise as coefficients of the rising factorial (see   
\cite[Identity 6.11]{concrete_1994})
\begin{equation}
\label{eq:stirling3}
x^{\overline{m}}=\sum_{l_2}{ m\brack l_2}x^{l_2}.
\end{equation}

A crucial observation is the following identity, which will be useful in the asymptotic analysis of Algorithm \ref{alg_one} when interference-connectivity distances $s,v$ are equal to $\frac{1}{\lambda}.$
\begin{Lemma}
\label{lem:sum_tec}
Assume that $a$ is an even positive number. Then
$$\sum_{j}\binom{a}{j}(-1)^j{j\brack j-k}=\begin{cases} 0 &\mbox{if } 2k <a \\
\frac{a!}{\left(\frac{a}{2}\right)!2^{\frac{a}{2}}} & \mbox{if } 2k= a. \end{cases}
$$
\end{Lemma}
\begin{Remark}
The following Mathematica code can be used to confirm numerically the validity
of Lemma \ref{lem:sum_tec}.
\begin{verbatim}
F[a_,k_]:=Sum[Binomial[a,j]*(-1)^j
*StirlingS1[j,j-k],{j,k,a}]
\end{verbatim}
Then the following command
\begin{verbatim}
F[a,k] 
\end{verbatim}
gives the result of Lemma \ref{lem:sum_tec} for fixed parameters $a$ and $k$.
\end{Remark}
We will also use many times Jensen's inequality for expectactions. If $f$ is a convex function, then
\begin{equation}
\label{eq:jensen}
f\left( \mathbf{E}[X]\right)\le  \mathbf{E}\left[f(X)\right]
\end{equation}
provided the expectations exists (see \cite[Proposition 3.1.2]{ross_2002}).

The following inequality for the general random variable will be useful in the threshold tight bounds when interference-connectivity distances are greater or equal to $\frac{1}{\lambda}.$

If $Z$ is the random variable such that  $\mathbf{E}[|Z|]< \infty$ and $q\in \mathbb{R},$ then
\begin{equation}
\label{eq:first101}
\left|\mathbf{E}[Z]-q\right|\le \mathbf{E}[|Z-q|]\le\ \mathbf{E}[|Z|]+|q|.
\end{equation}
Notice that the left side of Inequality (\ref{eq:first101})  is the special case of Jensen's inequality for $X:=Z-q$ and $f(x)=|x|.$
The right side of Inequality (\ref{eq:first101}) follows from the triangle inequality $|Z-q|\le |Z|+|q|$ and
the monotonicity of the expected value.

We will also use the following notation
\begin{equation}
\label{eq:positive_part}
|x|^{+}=\max\{x,0\}
\end{equation}
for positive parts of $x\in\mathbb{R}.$



Let $f$ be non-negative integer. Then
\begin{equation}
\label{id:sumowanie}
\sum_{i=2}^n (i-1)^f =\frac{1}{f+1}n^{f+1}+\sum_{l=0}^{f}c_ln^l,
\end{equation}
where $c_l$ are some constants independent on $n$ (see \cite[Formula (6.78)]{concrete_1994}).

\section{Sensors on the Line}
\label{chap:line}
Fix $a\ge 1.$ 
Let us recall that $\epsilon, \tau>0$ are arbitrary small constants independent on $n$ and $\lambda.$
In this section we analyze $(s,v)-IP$ interference-connectivity problem when the $n$ sensors are placed in the half-infinite interval $[0,\infty)$
according to Poisson process with arrival rate $\lambda.$
\subsection{Analysis of Algorithm \ref{alg_one}}
In this subsection we present and analyse \textit{asymptocically optimal} algorithm $MV_1(n,s)$\footnote{ We note that asymptotic analysis of Algorithm \ref{alg_one}
is \textbf{crucial} in deriving \textit{the threshold phenomena}.}
(see Algorithm \ref{alg_one}).
%
\begin{algorithm}[H]
\caption{$MV_1(n,s)\,\,$ Moving sensors in the $[0,\infty)$; $s>0.$}
\label{alg_one}
\begin{algorithmic}[1]
 \REQUIRE The initial location $X_1\le X_2\le \dots \le X_n$ of the $n$ sensors in the $[0,\infty)$ 
 according to Poisson process with arrival rate $\lambda.$
 \ENSURE  The final positions of the sensors such that each pair of consecutive sensors is separated by the distance equal to $s.$
 \FOR{$i=2$  \TO $n$ } 
 \STATE{move the sensor $X_{i}$ at the position $X_1+(i-1)s;$}
 \ENDFOR
\end{algorithmic}
\end{algorithm}
 
We prove the following tight bound.
\begin{theorem} 
\label{thm:mainexactodd} Fix $\Delta\ge 0$ independent on $n$ and $\lambda.$
Let $a$ be an even natural number. 
The expected $a-$total movement of algorithm $MV_1\left(n,\frac{1+\Delta}{\lambda}\right)$
is respectively 
$$
\begin{cases}
\frac{a!}{2^{\frac{a}{2}}\left(\frac{a}{2}+1\right)!}\frac{n^{1+\frac{a}{2}}}{\lambda^a} +\frac{O\left(n^{\frac{a}{2}}\right)}{\lambda^a}
\,\,\,&\text{when}\,\,\,\Delta=0,\\
\frac{\Theta\left(n^{1+a}\right)}{\lambda^a}\,\,\,\,\,\,\,\,\,\,\,\,\,\,\,\,\,\,\,\,\,\,\,\,\,\,\,\,\,\,\,\,&\text{when}\,\,\,\Delta>0.
\end{cases}
$$

\end{theorem}
\begin{proof}
Let $X_i$ be the arrival time of the $i-$th event in a Poisson process with arrival rate $\lambda$, i.e. the position of the $i-$th sensor in the $[0,\infty).$
We know that the random variable $X_i-X_1=X_{i-1}$ obeys the Gamma distribution with density
$$f_{i-1,\lambda}(t)=\lambda e^{-\lambda t}\frac{(\lambda t)^{i-2}}{(i-2)!}$$
for $i=2,3,\dots ,n.$ (see Equation (\ref{eq:probadens}) for $j=1$ and $l=i-1$).
Assume that, $a$ is even natural number.
Let $D^{(a)}_i$ be the expected distance to the power $a$ between $X_i-X_1$ and the $i-$th sensor position, $t_{i-1}=(1+\Delta)\frac{i-1}{\lambda},$  hence given by

\begin{align*}
D^{(a)}_i&=\int_{0}^{\infty}|t-t_{i-1}|^af_{i-1,\lambda}(t)dt\\
&=\int_{0}^{\infty}(t_{i-1}-t)^af_{i-1,\lambda}(t)dt.
\end{align*}
Observe that
$$
D^{(a)}_i=\sum_{j}\binom{a}{j}\left((1+\Delta)\frac{i-1}{\lambda}\right)^{a-j}(-1)^j\int_{0}^{\infty}t^jf_{i-1,\lambda}(t)dt.
$$
Using (\ref{integral_2}) we see that
$$
D^{(a)}_i=\frac{1}{\lambda^a}\sum_{j}\binom{a}{j}(1+\Delta)^{a-j}(-1)^j(i-1)^{a-j}\frac{(i+j-2)!}{(i-2)!}.     
$$
Let $j\in\{0,\dots, a\}.$
Applying Identity (\ref{eq:stirling3}) we deduce that
\begin{align*}
(i-1)^{a-j}\frac{(i+j-2)!}{(i-2)!}&=(i-1)^{a-j}(i-1)^{\overline{j}}\\
=\sum_{k} {j\brack j-k}(i-1)^{a-k}.
\end{align*}
Hence
$$
D^{(a)}_i=\frac{1}{\lambda^a}\sum_{j}\sum_{k}\binom{a}{j}(1+\Delta)^{a-j}(-1)^j (i-1)^{a-k}{j\brack j-k}.
$$
Changing the summation we get
$$
D^{(a)}_i=\frac{1}{\lambda^a}\sum_{k}(i-1)^{a-k}\sum_{j}\binom{a}{j}(1+\Delta)^{a-j}(-1)^j{j\brack j-k}.
$$
Now, we will estimate separately when $\Delta=0$ and when $\Delta>0.$

\textit{Case $\Delta=0.$}
Applying Lemma \ref{lem:sum_tec} we get
$$
D^{(a)}_i=\frac{1}{\lambda^a}\frac{a!}{\left(\frac{a}{2}\right)!2^{\frac{a}{2}}}\cdot (i-1)^{\frac{a}{2}}+\frac{1}{\lambda^a}\sum_{2k> a}C_{1,a-k}\cdot (i-1)^{a-k},
$$
where $C_{1,a-k}$ depends only on $a$ and $k.$
Using Identity (\ref{id:sumowanie}) 
we conclude that the expected $a-$total movement of algorithm $MV_1\left(n,\frac{1}{\lambda}\right)$ is
$$
\sum_{i=2}^{n}D^{(a)}_i=\frac{a!}{2^{\frac{a}{2}}\left(\frac{a}{2}+1\right)!}\frac{n^{1+\frac{a}{2}}}{\lambda^a} +\frac{O\left(n^{\frac{a}{2}}\right)}{\lambda^a}.
$$
This is enough to prove the case when $\Delta=0.$

\textit{Case $\Delta>0.$}
Observe that $\sum_{j}\binom{a}{j}(-1)^j(1+\Delta)^{a-j}{j\brack j}=\Delta^a.$ Therefore
$$
D^{(a)}_i=\frac{\Delta^a}{\lambda^a}\cdot (i-1)^{a}+\frac{1}{\lambda^a}\sum_{k> 0}C_{2,a-k}\cdot (i-1)^{a-k},
$$
where $C_{2,a-k}$ depends only on $a$ and $k.$
Again, applying Identity (\ref{id:sumowanie})
we conclude that the expected sum
of displacements to the power $a$ of algorithm $MV_1\left(n,\frac{1+\Delta}{\lambda}\right)$ is
$
\frac{\Theta\left(n^{1+a}\right)}{\lambda^a}.
$
This is enough to prove the case when $\Delta>0$
which completes the proof of Theorem \ref{thm:mainexactodd}.
\end{proof}
The next theorem extends the case $\Delta=0$   in our Theorem \ref{thm:mainexactodd} to real valued exponents.

It is worthwhile to mention that, the asymptotic result of Theorem \ref{thm:mainexactodd2} for all exponents $a\ge 1$ follows from 
Theorem \ref{thm:mainexactodd} when $a$ is positive even natural
and the probabilistic representation of absolute moments in terms of characteristic functions (see \cite{ushakov}, \cite{bahr} for details).
\begin{theorem} 
\label{thm:mainexactodd2} 
Fix $a\ge 1.$  
The expected $a$-total movement of algorithm $MV_1\left(n,\frac{1}{\lambda}\right)$
is respectively \footnote{The Gamma function $\Gamma(z)$ is the extension of the factorial to positive real number arguments. When $a$ is an even natural number
we have $\Gamma\left(\frac{a}{2}+2\right)=\left(\frac{a}{2}+1\right)!.$}
$$
\frac{a!}{2^{\frac{a}{2}}\Gamma\left(\frac{a}{2}+2\right)}\frac{{n^{1+\frac{a}{2}}}}{\lambda^a} +\frac{O\left({n^{\frac{a}{2}}}\right)}{\lambda^a},
$$
\end{theorem}
We also prove the following tight bound for $1-$total movement of Algorithm \ref{alg_one} when $s=\frac{1+\epsilon}{\lambda}$ and $\epsilon>0.$
\begin{theorem} 
\label{thm:mainexact_epsilon} 
Let $\epsilon>0$ be a constant independent on $n$ and $\lambda.$
Then the expected $1-$total movement
of algorithm $MV_1\left(n,\frac{1+\epsilon}{\lambda}\right)$
is respectively 
$
\frac{\Theta\left(n^{2}\right)}{\lambda}.
$
\end{theorem}
\begin{proof}
 Let $D^{(1)}_i$ be the expected distance between $X_i-X_1$ and the $i-$th sensor position, $t_{i-1}=(1+\epsilon)\frac{i-1}{\lambda}$ for $i=2,3,\dots ,n,$  hence given by
$$
D^{(1)}_i=\int_{0}^{\infty}|t-t_{i-1}|f_{i-1,\lambda}(t)dt.
$$
Let us recall that $\mathbf{E}\left[X_i-X_1\right]=\frac{i-1}{\lambda}$ (see (\ref{integral_2}) for $l=i-1$).
Applying Inequality (\ref{eq:first101}) for $Z=X_i-X_1,  q=t_{i-1}$ and $\epsilon>0$ we have
$$
\epsilon\frac{i-1}{\lambda}
\le D^{(1)}_i \le(2+\epsilon)\frac{i-1}{\lambda}.
$$
Since $\sum_{i=2}^{n}(i-1)=\frac{(n-1)n}{2},$ we have
\begin{equation}
\label{eq:delta102} 
\sum_{i=2}^{n}D^{(1)}_i =\frac{\Theta\left(n^{2}\right)}{\lambda}
\end{equation}
This completes the proof of Theorem \ref{thm:mainexact_epsilon}. 
\end{proof}

\subsection{Expected Minimal $a-$total Movement for $s=v=\frac{1}{\lambda}$}
\label{subsection:tight}
In this subsection we look at the expected minimal $a-$total movement when the interference-connectivity distances $s$ and $v$ are equal to $\frac{1}{n}.$
We prove the upper bound $\frac{O\left(n^{1+\frac{a}{2}}\right)}{\lambda^a}$ on the expected minimal $a-$total movement (see Theorem \ref{thm:mainexactfract}).
Our Theorem \ref{thm:mainexactfractlower} and Theorem \ref{thm:mainexactfractlower1a} give the lower bound $\frac{\Omega\left(n^{1+\frac{a}{2}}\right)}{\lambda^a}$ on the expected minimal $a-$total movement.

We begin with a theorem which indicates how to apply the results of Theorem \ref{thm:mainexactodd} to the upper bound on the expected $a-$total movement, when $a\ge 1.$
In the proof of Theorem \ref{thm:mainexactfract}, we combine together discrete H\"older inequality, Jensen's inequality and the asymptotic result of Theorem
\ref{thm:mainexactodd}.
\begin{theorem} 
\label{thm:mainexactfract} 
Let $a\ge 1.$ 
The expected $a-$total movement
of algorithm $MV_1\left(n,\frac{1}{\lambda}\right)$
is respectively $\frac{O\left(n^{1+\frac{a}{2}}\right)}{\lambda^a}.$
\end{theorem}
\begin{proof}
Assume that $a\ge 1.$ 
Let $D^{(a)}_i$ be the expected distance to the power $a$ between $X_i-X_1$ and the $i^{th}$ sensor position. 
Let $b$ be the even natural number such that $b-a>0.$
Then we use discrete H\"older inequality with parameters $\frac{b}{a}$ and $\frac{b}{b-a}$
and get
\begin{eqnarray}
\sum_{i=2}^{n}D^{(a)}_i
&\le& \notag
\left(\sum_{i=2}^{n}\left(D^{(a)}_i\right)^{\frac{b}{a}}\right)^{\frac{a}{b}}
\left(\sum_{i=2}^{n}1\right)^{\frac{b-a}{b}  }\\
&=&  \label{eq:holder1}
\left(\sum_{i=2}^{n}\left(D^{(a)}_i\right)^{\frac{b}{a}}\right)^{\frac{a}{b}}
(n-1)^{\frac{b-a}{b}}.
\end{eqnarray}
Next we use Jensen's inequality (see (\ref{eq:jensen})) for $f(x)=x^{\frac{b}{a}}$ and $\mathbf{E}[X]=D^{(a)}_i$
and get
\begin{equation}
\label{eq:jensen_b2}
\left(D^{(a)}_i\right)^{\frac{b}{a}}\le D^{(b)}_i.
\end{equation}
Combining  together (\ref{eq:holder1}), (\ref{eq:jensen_b2}) and Theorem \ref{thm:mainexactodd} we deduce that
$$
\sum_{i=2}^{n}D^{(a)}_i\le
\left(\Theta\left(\frac{n^{1+\frac{b}{2}}}{\lambda^{b}}\right)\right)^{\frac{a}{b}}(n-1)^{\frac{b-a}{b}}=\frac{\Theta\left(n^{1+\frac{a}{2}}\right)}{\lambda^a}.
$$
This is enough to prove the upper bound which finishes the proof of Theorem \ref{thm:mainexactfract}.
\end{proof}
We now prove the desired lower bound for expected $1-$total movement.
\begin{theorem} 
\label{thm:mainexactfractlower} 
Any sensor's displacement algorithm which solves $\left(\frac{1}{\lambda},\frac{1}{\lambda}\right)-IP$ problem
requires expected $1-$total movement
of at least
$\frac{\Omega\left(n^{\frac{3}{2}}\right)}{\lambda}.$ 
\end{theorem}
\begin{proof}
Before providing the proof of the theorem we make two important observations.

Let $X_1<X_2<\dots <X_n$ be the initial positions of the sensors.
Recall that by the monotonicity lemma no sensor $X_i$ is ever placed before sensor $X_j,$ for all $i<j.$

We assume that the final position of the first sensor is
the nonnegative random variable $Z$ with $\mathbf{E}[Z]<\infty$.

We are now ready to prove the theorem.
Let $X_i$ be the arrival times of the $i-$th event in Poisson process with arrival rate $\lambda.$ Let $t_i=\frac{i}{\lambda},$ for $i=1,2,\dots,n.$
Putting together Theorem \ref{thm:mainexactodd2} and Equation (\ref{eq:probadens}) for $j=1, l=i-1$ we have the following tight asymptotic result
\begin{equation}
\label{eq:cocon}
\sum_{i=1}^{n-1}\mathbf{E}\left[|X_i-t_i|\right]=C_1\frac{n^{\frac{3}{2}}}{\lambda}+\frac{O\left(n^{\frac{1}{2}}\right)}{\lambda},
\end{equation}
where $C_1=(\sqrt{2}\Gamma(5/2))^{-1}.$ 
There are two cases to consider.

\textbf{Case 1.} The algorithm moves the sensor $X_i$ to the position $Z+b_i,$ where  $b_i=\frac{1}{\lambda}(i-1),$ for $i=1,2,\dots n$ and 
$\mathbf{E}[Z]>\frac{1}{2}C_1\frac{n^{\frac{1}{2}}}{\lambda}.$

Combining together Inequality (\ref{eq:first101}) 
for $Z:=X_i-Z$,  $q = b_i$ Equation (\ref{integral_2}) for $l=i$
and the triangle inequality we get
\begin{align*}
\sum_{i=1}^n\mathbf{E}\left[\left|X_{i}-Z-b_i\right|\right]&\ge \sum_{i=1}^n\left|\mathbf{E}[Z]-\frac{1}{\lambda}\right|\\
&\ge\sum_{i=1}^n\left(\left|\frac{1}{2}C_1\frac{n^{\frac{1}{2}}}{\lambda}\right|-\left|\frac{1}{\lambda}\right|\right)=\frac{\Theta\left(n^{\frac{3}{2}}\right)}{\lambda}.
\end{align*}
This is enough to prove the first case.

\textbf{Case 2.} The algorithm moves the sensor $X_i$ to the position $Z+b_i,$ where  $b_i=\frac{1}{\lambda}(i-1),$ for $i=1,2,\dots n$ and $\mathbf{E}[Z]\le\frac{1}{2}C_1\frac{\sqrt{n}}{\lambda}.$

Let us recall that $t_i=\frac{i}{\lambda},$ for $i=1,2,\dots,n.$
Using the triangle inequality
$$|X_{i}-t_{i}|\le|X_{i}-(Z+b_i)|+|(Z+b_i)-t_{i}|$$
we get
$$
\sum_{i=1}^n\mathbf{E}\left[|X_{i}-(Z+b_i)|\right]
\ge\sum_{i=1}^n\mathbf{E}\left[|X_{i}-t_{i}|\right]-\sum_{i=1}^n\mathbf{E}\left|Z-\frac{1}{\lambda}\right|.
$$
Combining together Equation (\ref{eq:cocon}), assumption $\mathbf{E}[Z]\le\frac{1}{2}C_1\frac{\sqrt{n}}{\lambda}$ and the triangle inequality
$\mathbf{E}\left[\left|Z-\frac{1}{n}\right|\right]\le \mathbf{E}[Z]+\frac{1}{\lambda}$  we have
$$\sum_{i=1}^n\mathbf{E}\left[|X_{i}-(Z+b_i)|\right]\in\frac{\Omega\left(n^{\frac{3}{2}}\right)}{\lambda}.$$
This is enough to prove the second case
and sufficient to complete the proof of Theorem \ref{thm:mainexactfractlower}.
\end{proof}
We now apply Theorem \ref{thm:mainexactfractlower} in order to derive the lower bound on the expected $a-$total movement when $a>1.$
Let us note that the proof of Theorem \ref{thm:mainexactfractlower1a}  is analogous to the proof of Theorem \ref{thm:mainexactfract} .
\begin{theorem} 
\label{thm:mainexactfractlower1a} 
Let $a> 1.$ 
Then any sensor's displacement algorithm which solves $\left(\frac{1}{\lambda},\frac{1}{\lambda}\right)-IP$ problem
requires expected $a-$total movement
of at least
$\frac{\Omega\left(n^{1+\frac{a}{2}}\right)}{\lambda^a}.$ 
\end{theorem}

\subsection{Expected Minimal $a-$total Movement for $\frac{1}{\lambda}<s\le v$}
\label{subsection:linear}
In this subsection we study the expected minimal $a-$total movement when the interference-connectivity distances $s$ and $v$ are greater than $\frac{1}{\lambda}.$
We give the upper bound $\frac{O\left(n^{1+a}\right)}{\lambda^a}$ on the expected $a-$total movement, when $a> 1$ (see Theorem \ref{thm:mainexactfract_epsilon})
and the lower bound $\frac{\Omega\left(n^{1+a}\right)}{\lambda^a}$  on the expected $a-$total movement, when $a\ge 1$ (see Theorem \ref{thm:mainexactfractlower_eps} and Theorem \ref{thm:mainexactfractlower_eps1}).

We begin with a theorem which indicates how to apply the results of Theorem \ref{thm:mainexactodd} to the upper bound on the expected $a-$total movement, when $a>1.$
\begin{theorem} 
\label{thm:mainexactfract_epsilon} 
Let $\epsilon>0$ be a constant independent on $n$ and $\lambda.$ Let $a>1.$ 
The expected $a-$total movement of algorithm $MV_1\left(n,\frac{1+\epsilon}{\lambda}\right)$
is in
$
\frac{O\left(n^{1+a}\right)}{\lambda^a}. 
$
\end{theorem}
We can now prove the desired lower bound for expected $1-$total movement. We note that the proof of Theorem \ref{thm:mainexactfractlower_eps}   is analogous to the proof of Theorem \ref{thm:mainexactfractlower}.
\begin{theorem} 
\label{thm:mainexactfractlower_eps} 
Fix $\tau\ge \epsilon>0$ independent on $n$ and $\lambda.$
Then any sensor's displacement algorithm 
which solves $\left(\frac{1+\epsilon}{\lambda},\frac{1+\tau}{\lambda}\right)-IP$ problem
requires
expected $1-$total movement
of at least $\frac{\Omega\left(n^2\right)}{\lambda}.$ 
\end{theorem}
We now apply Theorem \ref{thm:mainexactfractlower_eps} in order to derive the lower bound on the expected $a-$total movement when $a>1.$
\begin{theorem} 
\label{thm:mainexactfractlower_eps1} 
Fix $\tau\ge \epsilon>0$ independent on $n$ and $\lambda.$ Let $a> 1.$ 
Then any sensor's displacement algorithm 
which solves $\left(\frac{1+\epsilon}{\lambda},\frac{1+\tau}{\lambda}\right)-IP$ problem
requires
expected $a-$total movement
of at least $\Omega(n).$ 
\end{theorem}
\subsection{Expected Minimal $a-$total Movement for $s<\frac{1}{\lambda}$ and $v>\frac{1}{\lambda}$}
\label{subsection:sublinear}
In this subsection we give algorithm $I_1(n,s,v)$ (see Algorithm \ref{alg_interference}) to solve $(s,v)-IP$ interference-connectivity problem.
Let $\epsilon, \tau>0$ be constants independent on $n$ and $\lambda.$ Let $a\ge 1.$
We show that expected $a-$total movement of algorithm
$I_1\left(n,\frac{1-\epsilon}{\lambda},\frac{1+\tau}{\lambda}\right)$
is in
$\frac{O\left(n\right)}{\lambda^a}.$
\begin{algorithm}[H]
\caption{$I_1(n,s,v)$ Moving sensors in the $[0,\infty)$; $0<s<v.$} 
\label{alg_interference}
\begin{algorithmic}[1]
 \REQUIRE The initial location $X_1\le X_2\le \dots \le X_n$ of the $n$ sensors in the $[0,\infty)$ according to Poisson process with arrival rate $\lambda.$ 
 \ENSURE  The final positions of the sensors such that\\ {\large$\forall_{i=2,3\dots ,n}$}
 $v\ge X_i-X_{i-1}\ge s$ (so as that the distance between consecutive sensors is less or equal $v$ and greater or equal to $s$).
 \FOR{$i=2$  \TO $n$ } 
 \IF{$X_i-X_{i-1}<s$}
 \STATE{move left-to-right the sensor $X_i$ at the new position $s+X_{i-1};$}
 \ELSIF{$X_i-X_{i-1}>v$}
 \STATE{move right-to-left the sensor $X_i$ at the new position $v+X_{i-1};$}
 \ELSE
 \STATE{do nothing;}
 \ENDIF
 \ENDFOR
\end{algorithmic}
\end{algorithm}
The next theorem gives the desired upper bound. 
\begin{theorem}
\label{thm:interfere} 
Fix $\epsilon,\tau>0$ independent on $n$ and $\lambda.$ Let $a\ge 1.$ 
The expected $a-$total movement of algorithm $I_1\left(n,\frac{1-\epsilon}{\lambda},\frac{1+\tau}{\lambda}\right)$
is in
$\frac{O\left(n\right)}{\lambda^a}.$
\end{theorem}
The general strategy of our proof of Theorem \ref{thm:interfere} is the following. Firstly we estimate the expected $a-$total movement of algorithm $I_1(n,s,v)$ by the sum
$$
\sum_{l=1}^{n}\frac{n}{l} \mathbf{E}\left[(|sl-X_{l}|^{+})^a\right]+\sum_{l=1}^{n}\frac{n}{l} \mathbf{E}\left[(|X_{l}-vl|^{+})^a\right].\footnote{This estimation is valid for general class of distribution with property
(\ref{eq:probadens}).} 
$$
Then the result of Theorem \ref{thm:interfere} is a consequence of the following lemma.
\begin{Lemma}
 \label{lemma_d}
 Fix $\epsilon>0$ independent on $n$ and $\lambda.$
Let $a \ge 1$ and let $s=\frac{1-\epsilon}{\lambda},$ $v=\frac{1+\tau}{\lambda}.$ Assume that random variable $X_l$ obeys Gamma distribution with parameters $l\in\mathbf{N}\setminus\{0\}$ and $\lambda>0.$
Then
\begin{equation}
\label{eq:lem_d1}
\sum_{l=1}^{n}\frac{n}{l} \mathbf{E}\left[(|sl-X_{l}|^{+})^a\right]=\frac{O\left(n\right)}{\lambda^a},
\end{equation}
\begin{equation}
\label{eq:lem_d2}
\sum_{l=1}^{n}\frac{n}{l} \mathbf{E}\left[(|X_{l}-vl|^{+})^a\right]=\frac{O\left(n\right)}{\lambda^a}.
\end{equation}
\end{Lemma}
\section{Sensors in the Higher Dimension.}
\label{chap:higher}
Fix $d\in\mathbb{N}\setminus\{0,1\}$ and $a\ge 1.$ Let $n=m^d$ for some $m\in\mathbb{N}.$

Let us recall that $\epsilon, \tau>0$ are arbitrary small constants independent on $n$ and $\lambda.$

We define our random placement and movement as follows.

\begin{Definition}[reallocation in $[0,\infty)^d$]
\label{def:atotek}
Consider $n$ sensors that are randomly placed  in the hyperoctant $[0,\infty)^d$ according to $d$ identical and independent Poisson processes
$X^{(1)}_i, X^{(2)}_i,\dots,X^{(d)}_i,$ for $i=1,2,\dots, n^{1/d}$
each with arrival rate $\lambda.$ 
\begin{itemize}
 \item The position of a sensor in the $\mathbb{R^d_+}$ is determined by the $d$ coordinates
$(X^{(1)}_{i_1}, X^{(2)}_{i_2},\dots,X^{(d)}_{i_d}),$ where $1\le i_1,i_2,\dots , i_d\le n^{1/d}.$
\item We have initially $n^{(d-1)/d}$ rows and $n^{(d-1)/d}$ columns such that each column and each row has $n^{1/d}$ random sensors.
\end{itemize}
Then, we reallocate the random sensors so as:
\begin{enumerate}
\item[(a)] the sensors move  only along the axes,
 \item [(b)] we have finally  $n^{(d-1)/d}$ rows and $n^{(d-1)/d}$ columns such that each column and each row has $n^{1/d}$ sensors,
 \item[(c)] in each column and in each row the sensors satisfy $(s,v)-IP$ interference-connectivity requirement.
\end{enumerate}
\end{Definition}
Figure 1 illustrates our initial random displacement of sensors in two dimensions.
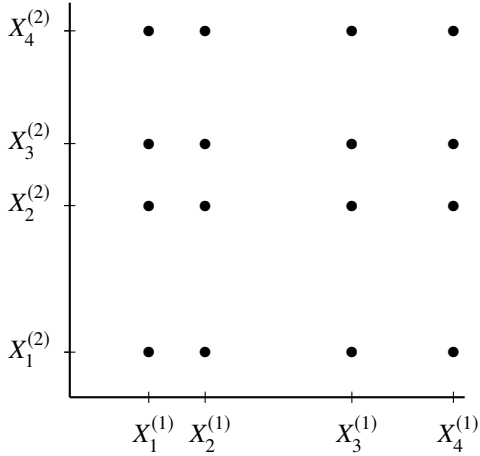
\begin{figure}[H]
\label{rand:pr}
\setlength{\unitlength}{0.75mm}
\centering
\begin{picture}(80,80)(-10,-10)

\put(0,0){\thicklines{\line(0,1){70}}}
\put(14,8){\circle*{2}}

\put(14,34){\circle*{2}}

\put(14,45){\circle*{2}}

\put(14,65){\circle*{2}}

\put(24,8){\circle*{2}}

\put(24,34){\circle*{2}}

\put(24,45){\circle*{2}}

\put(24,65){\circle*{2}}

\put(50,8){\circle*{2}}

\put(50,34){\circle*{2}}

\put(50,45){\circle*{2}}

\put(50,65){\circle*{2}}

\put(68,8){\circle*{2}}

\put(68,34){\circle*{2}}

\put(68,45){\circle*{2}}

\put(68,65){\circle*{2}}

\put(0,8){\line(-1,0){1}}
\put(1,8){\line(-1,0){2}}
\put(-11,7){$X^{(2)}_{1}$}

\put(0,34){\line(-1,0){1}}
\put(1,34){\line(-1,0){2}}
\put(-11,33){$X^{(2)}_{2}$}

\put(0,45){\line(-1,0){1}}
\put(1,45){\line(-1,0){2}}
\put(-11,44){$X^{(2)}_{3}$}

\put(0,65){\line(-1,0){1}}
\put(1,65){\line(-1,0){2}}
\put(-11,64){$X^{(2)}_{4}$}

\put(0,0){\thicklines{\line(1,0){70}}}

\put(14,0){\line(0,-1){1}}
\put(14,0){\line(0,0){1}}
\put(11,-8){$X^{(1)}_{1}$}

\put(24,0){\line(0,-1){1}}
\put(24,0){\line(0,0){1}}
\put(21,-8){$X^{(1)}_{2}$}

\put(50,0){\line(0,-1){1}}
\put(50,0){\line(0,0){1}}
\put(47,-8){$X^{(1)}_{3}$}

\put(68,0){\line(0,-1){1}}
\put(68,0){\line(0,0){1}}
\put(65,-8){$X^{(1)}_{4}$}

\end{picture}
\caption{The mobile sensors located in the quadrant $[0,\infty)^2$ according to 2 identical and independent Poisson processes.}
\label{fig:aser}
\end{figure}

In order to fullfil the requirements (a), (b), (c) two algorithms are presented. Namely,
\begin{itemize}
 \item for the case of $s=\frac{1}{\lambda},$ $v=\frac{1}{\lambda}$ we show that the expected $a-$total movement of algorithm
  $MV_d\left(n,\frac{1}{\lambda}\right)$ is in $\frac{\Theta\left(n^{1+\frac{a}{2d}}\right)}{\lambda^a}$ (see Theorem  \ref{thm:mainexacthigh}),
 \item for the case of $s=\frac{1-\epsilon}{\lambda},$ $v=\frac{1+\tau}{\lambda}$ we prove that the expected $a-$total movement of algorithm $I_d\left(n,\frac{1-\epsilon}{\lambda}, \frac{1+\tau}{\lambda} \right)$ is in
$
O\left(n^{1-\frac{a}{d}}\right)
$
(see Theorem \ref{thm:mainexacthigh_eps_minus}).
\end{itemize}
\begin{algorithm}[H]
\caption{$MV_d(n,s)\,\,$ Moving sensors in the $[0,\infty)^d,$ $d\ge 2,$ 
$s>0.$ }
\label{alg_d}
\begin{algorithmic}[1]
 \REQUIRE The initial location $(X^{(1)}_{i_1}, X^{(2)}_{i_2},\dots,X^{(d)}_{i_d})$ of the $n$ sensors in the $[0,\infty)^d,$ $1\le i_1,i_2,\dots , i_d\le n^{1/d}.$
 \ENSURE  The final positions of the sensors such that in each column and in each row the consecutive sensors are separated by the distance equal to $s.$
 \STATE{{\Large $\forall_{1\,\, \le\,\, i_1,i_2,\dots , i_d\,\, \le\,\,  n^{1/d}}$} $\,\,\,$ move the sensor at the location $(X^{(1)}_{i_1},X^{(2)}_{i_2},\dots, X^{(d)}_{i_d})$ 
 to the position $(X^{(1)}_{1+(i_1-1)s},X^{(2)}_{1+(i_2-1)s},\dots, X^{(d)}_{1+(i_d-1)s});$}
\end{algorithmic}
\end{algorithm}

\begin{algorithm}[H]
\caption{$I_d(n,s,v)\,\,$ Moving sensors in the $[0,\infty)^d,$ $d\ge 2,$ $0<s<v.$}
\label{alg_din}
\begin{algorithmic}[1]
 \REQUIRE The initial location $(X^{(1)}_{i_1}, X^{(2)}_{i_2},\dots,X^{(d)}_{i_d})$ of the $n$ sensors in the $[0,\infty)^d,$ $1\le i_1,i_2,\dots , i_d\le n^{1/d}.$
 \ENSURE  The final positions of the sensors such that in each column and in each row the sensors satisfy
$(s,v)−-IP$ interference-connectivity requirement.
 \STATE{For each column and row in the $[0,\infty)^d$ apply algorithm $I_1(n^{1/d},s,v);$}
\end{algorithmic}
\end{algorithm}
We call a move of a sensor a \textit{sliding move} if the final position of the sensor is either in the same row or column as its initial position.

In this section, we restrict the movement of sensors to a \textit{sliding movement}. The claim is justified in Lemma \ref{lem:crucial}  whose simple proof is omitted.
Such a reduction of the movement is indeed \textit{crucial} and \textit{reduces the displacement of sensors in the higher dimension to the displacement of sensors in the half-infinite interval
$[0,\infty).$}
\begin{Lemma}
\label{lem:crucial}
 The optimal reallocation of sensors which ensures the requirements (a), (b), (c) is a \textit{sliding movement}.
\end{Lemma}
The next simple lemma will be helpfull in estimating the upper bound in Theorems \ref{thm:mainexacthigh}-\ref{thm:mainexacthigh_eps_minus}. 
\begin{Lemma}
\label{lem:cruciana}
Fix $a\ge 1.$  Let $M$ be the sensor movement in $[0,\infty)^d.$ Assume that $M=\sum_{i=1}^{d} M_i,$
where $M_i$ is the sensor movement along the $i-$th fixed axis $[0,\infty).$ 
Then 
$$\mathbf{E}[M^a]\le C_{a,d}\sum_{i=1}^{d}\mathbf{E}[M_i^a],$$
where $C_{a,d}$ is some constant which depend only on fixed $a$ and $d.$
\end{Lemma}
We now embark to extend the results from Section \ref{chap:line}  to the high dimensions. We can prove the following sequences of Theorem.

The next theorem  clarifies how the interference-connectivity distances $s=\frac{1}{\lambda},$ $v=\frac{1}{\lambda}$ affect the expected minimal 
$a-$total movement.
\begin{theorem} 
\label{thm:mainexacthigh} 
Let $a\ge 1$ be a constant. 
Assume that $n$ sensors are placed in the $[0,\infty)^d$ according to $d$ independent identical Poisson processes,
each with arrival rate $\lambda$ and the reallocation of sensors ensures the requirements (a), (b), (c).
If the interference-connectivity distances $s=\frac{1}{\lambda},$ $v=\frac{1}{\lambda}$
then the expected minimal $a-$total movement is in
$
\frac{\Theta\left(n^{1+\frac{a}{2d}}\right)}{\lambda^a}.
$
\end{theorem}
\begin{proof}
First of all, we discuss the proof of the upper bound. By Theorem \ref{thm:mainexactfract} applied to
$n:=n^{1/d}$ and 
for $n^{(d-1)/d}$ columns and $n^{(d-1)/d}$ rows, as well as Lemma \ref{lem:cruciana}
we derive that the expected $a-$total movement of algorithm $MV_d\left(n,\frac{1}{\lambda}\right)$ is 
$$2 C_{a,d} n^{{(d-1)}/d}\frac{O\left(\left(n^{1/d}\right)^{1+\frac{a}{2}}\right)}{\lambda^a}=\frac{O\left(n^{1+\frac{a}{2d}}\right)}{\lambda^a}.$$
Next we prove the lower bound. Since the movement of sensors along the axes is a \textit{sliding move} to attain the interference-connectivity distances $s=\frac{1}{\lambda},$ $v=\frac{1}{\lambda}$
in the $[0,\infty)^d$ the sensors have to attain the interference-connectivity distances $s=\frac{1}{\lambda},$ $v=\frac{1}{\lambda}$ in each column and each row.
By Theorem \ref{thm:mainexactfractlower} and Theorem \ref{thm:mainexactfractlower1a} applied to
$n:=n^{1/d}$ and 
for $n^{(d-1)/d}$ columns 
we have the following lower bound
$$n^{{(d-1)}/d}\frac{\Omega\left(\left(n^{1/d}\right)^{1+\frac{a}{2}}\right)}{\lambda^a}=\frac{\Omega\left(n^{1+\frac{a}{2d}}\right)}{\lambda^a}.$$
This is sufficient to complete the proof of Theorem \ref{thm:mainexacthigh}.
\end{proof}
We now analyze the expected minimal 
$a-$total movement when the  interference-connectivity distances $s$ and $v$ are greater than $\frac{1}{\lambda}.$
The proof of the next theorem is analogous to the proof of Theorem \ref{thm:mainexacthigh}.
\begin{theorem} 
\label{thm:mainexacthigh_eps} Fix $\tau\ge \epsilon>0$ independent on $n.$
Let $a\ge 1$ be a constant. 
Assume that $n$ sensors are placed in the $[0,\infty)^d$ according to $d$ independent identical Poisson processes,
each with arrival rate $\lambda$ and the reallocation of sensors ensures the requirements (a), (b), (c).
If the interference-connectivity distances are equal to $s=\frac{1+\epsilon}{\lambda},$ $v=\frac{1+\tau}{\lambda},$
then the expected  minimal $a-$total movement is in 
$
\frac{\Theta\left(n^{1+\frac{a}{d}}\right)}{\lambda^a}.
$
\end{theorem}
The next theorem provides the expected minimal 
$a-$total movement for $s<\frac{1}{\lambda}$ and $v>\frac{1}{\lambda}.$
\begin{theorem} 
\label{thm:mainexacthigh_eps_minus} 
Fix $\epsilon, \tau>0$ independent on $n.$
Let $a\ge 1$ be a constant. 
Assume that $n$ sensors are placed in the $[0,\infty)^d$ according to $d$ independent identical Poisson processes,
each with arrival rate $\lambda$  and the reallocation of sensors ensures the requirements (a), (b), (c). 
The expected $a-$total movement of algorithm $I_d\left(n,\frac{1-\epsilon}{\lambda}, \frac{1+\tau}{\lambda} \right)$ is in
$
\frac{O\left(n\right)}{\lambda^a}. 
$
\end{theorem}
\section{Experimental Results}
\label{sec:sim}
In this section, we provide a set of experiments to illustrate how interference distance $s$ and connectivity distance $v$ impact the minimal expected $a-$total movement.

Namely, we implemented Algorithms (\ref{alg_one}-\ref{alg_interference}) in Wolfram Mathematica $10.0$ for $a=1$ and $a=2$ and $\lambda=n.$
\footnote{\textbf{It is not difficult to repeat the simulation from this section for all exponents $a\ge 1$ and to any parameter $\lambda>0,$ as well as for Algorithms (\ref{alg_d}-\ref{alg_din}) to visualize and confirm the threshold phenomena.}}

Figure 3 and 6 illustrate the expected $a-$total movement of Algorithm \ref{alg_one} when $s=v=\frac{1.1}{n}$ 
for the number of sensors $n\in\{1,2,\dots, 3000\}.$ Observe that $1-$total movement ${E}_{1}(n)$ and 
$2-$total movement ${E}_{2}(n)$ are in $\Theta(n).$

In Figure 4 and 7 the black dots are the numerical results of the expected $a-$total movement of Algorithm \ref{alg_one} for $s=\frac{1}{n}.$
The additional curves $\left\{\left(n,\frac{1}{\sqrt{2}\Gamma\left(\frac{5}{2}\right)}\sqrt{n}\right), 1\le n\le 3000 \right\},$ 
$\left\{\left(n,\frac{1}{2}\right), 1\le n\le 3000 \right\}$ represent the exact theoretical estimations (see Theorem
\ref{thm:mainexactodd2} for $a=1$ and Theorem \ref{thm:mainexactodd} for $a=2,\,\,$ $\,\,\Delta=0$).

The expected $a-$total movements of Algorithm \ref{alg_interference} for the parameters $s=\frac{0.4}{n}$ and $v=\frac{1.6}{n}$ and the number of sensors
$n\in\{1,2,\dots, 3000\}$ are depicted in Figure 2 and 5. It can be seen that  $1-$total movement  ${E}_{1}(n)$ is in $\Theta(1)$ and $2-$total movement ${E}_{2}(n)$ is in $\Theta\left(\frac{1}{n}\right).$ 

\begin{figure}[H]\label{fig.1a}
  \begin{center}
 \begin{minipage}[b]{0.43\textwidth} \centering \includegraphics[width=1.00\textwidth]{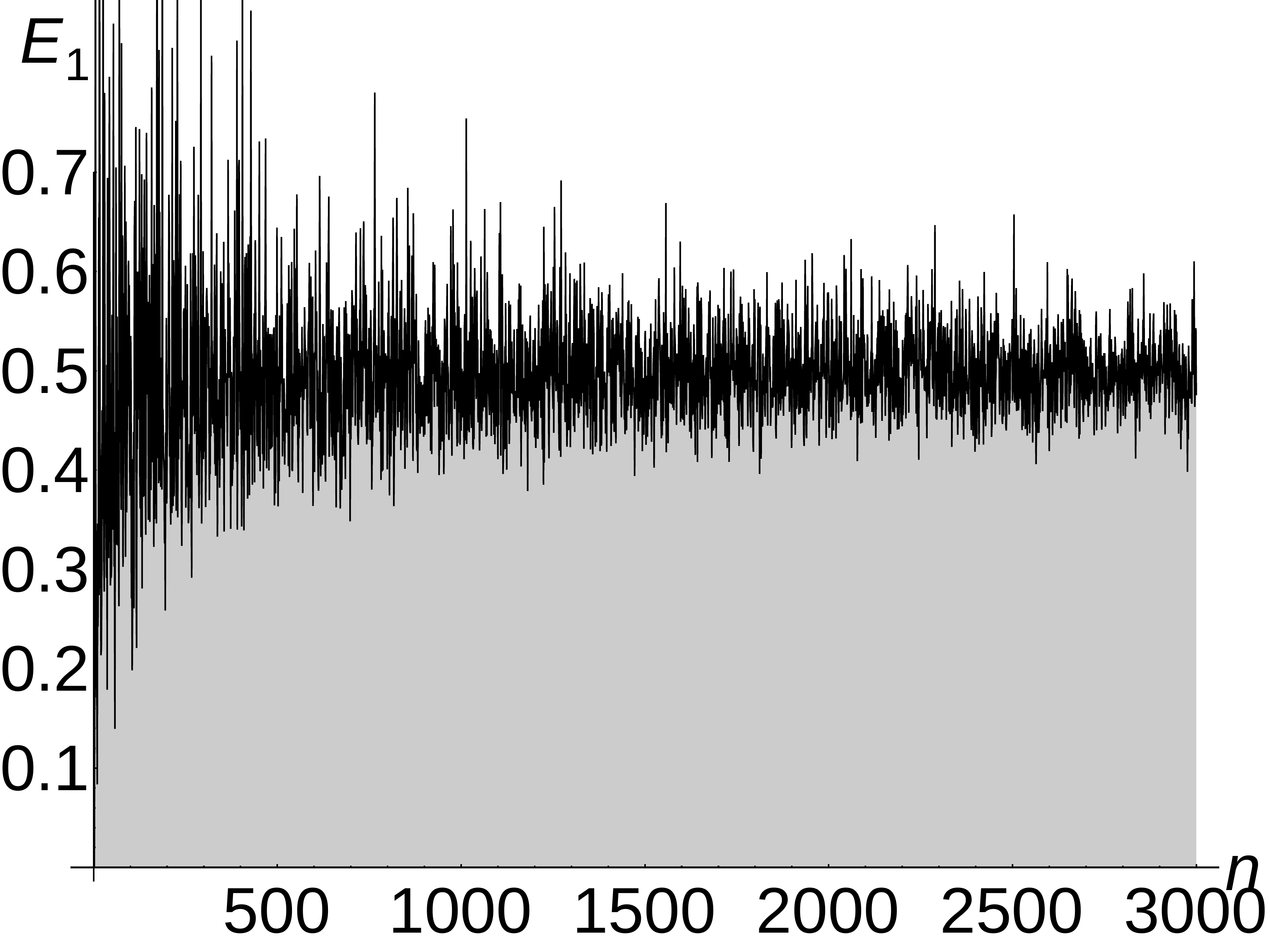}\\ ${E}_{1}(n)=\Theta(1)$\\ 
 \textbf{Case} $s<\frac{1}{n}$ \& $v>\frac{1}{n}.$
 \caption{The expected $1-$total movement ${E}_{1}(n)$ of \textbf{Algorithm \ref{alg_interference}} for $s=\frac{0.4}{n}$ and $v=\frac{1.6}{n}.$}\label{a:1}
 \end{minipage}
\end{center}
\end{figure}
\begin{figure}[H]\label{fig.1b}
  \begin{center}
 \begin{minipage}[b]{0.45\textwidth} \centering \includegraphics[width=1.00\textwidth]{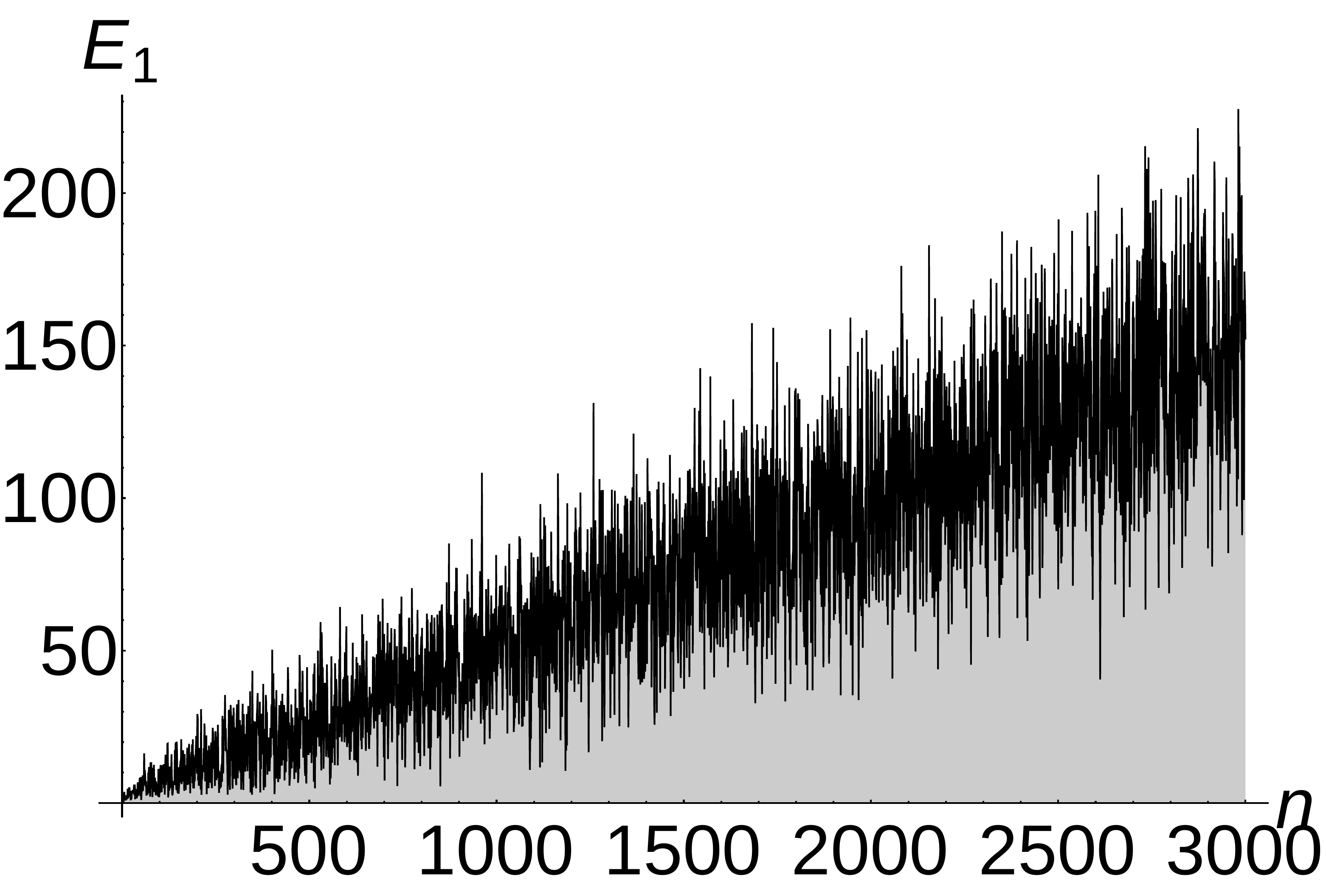}\\ ${E}_{1}(n)=\Theta(n)$ \\
 \textbf{Case} $s>\frac{1}{n}$ \& $v>\frac{1}{n}.$
  \caption{The expected $1-$total movement ${E}_{1}(n)$ of \textbf{Algorithm \ref{alg_one}} for $s=\frac{1.1}{n}$}\label{a:2}
 \end{minipage}
  \end{center}
 \end{figure}
 \begin{figure}[H]
\label{fig.1c}
  \begin{center}
    \includegraphics[width=0.43\textwidth]{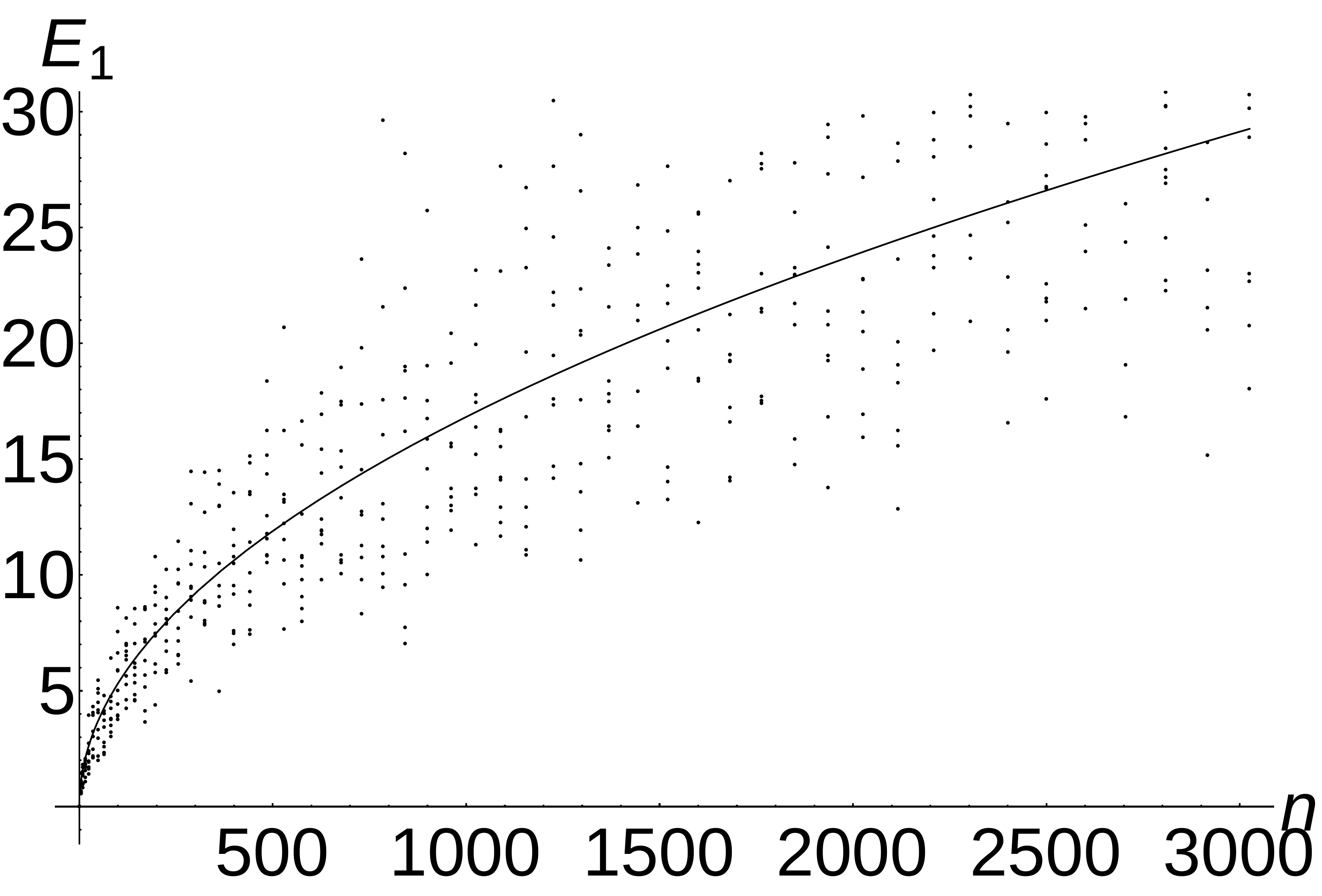}\\ ${E}_{1}(n)\sim \frac{1}{\sqrt{2}\Gamma\left(\frac{5}{2}\right)} \sqrt{n}$\\
     \textbf{Case} $s=\frac{1}{n}$ \& $v=\frac{1}{n}.$
  \end{center}
  \caption{The expected $1-$total movement ${E}_{1}(n)$ of \textbf{Algorithm \ref{alg_one}} for $s=\frac{1}{n}$}\label{a:3}
\end{figure}
\begin{figure}[H]
 \label{fig.2a}
  \begin{center}
 \begin{minipage}[b]{0.43\textwidth} \centering \includegraphics[width=1.00\textwidth]{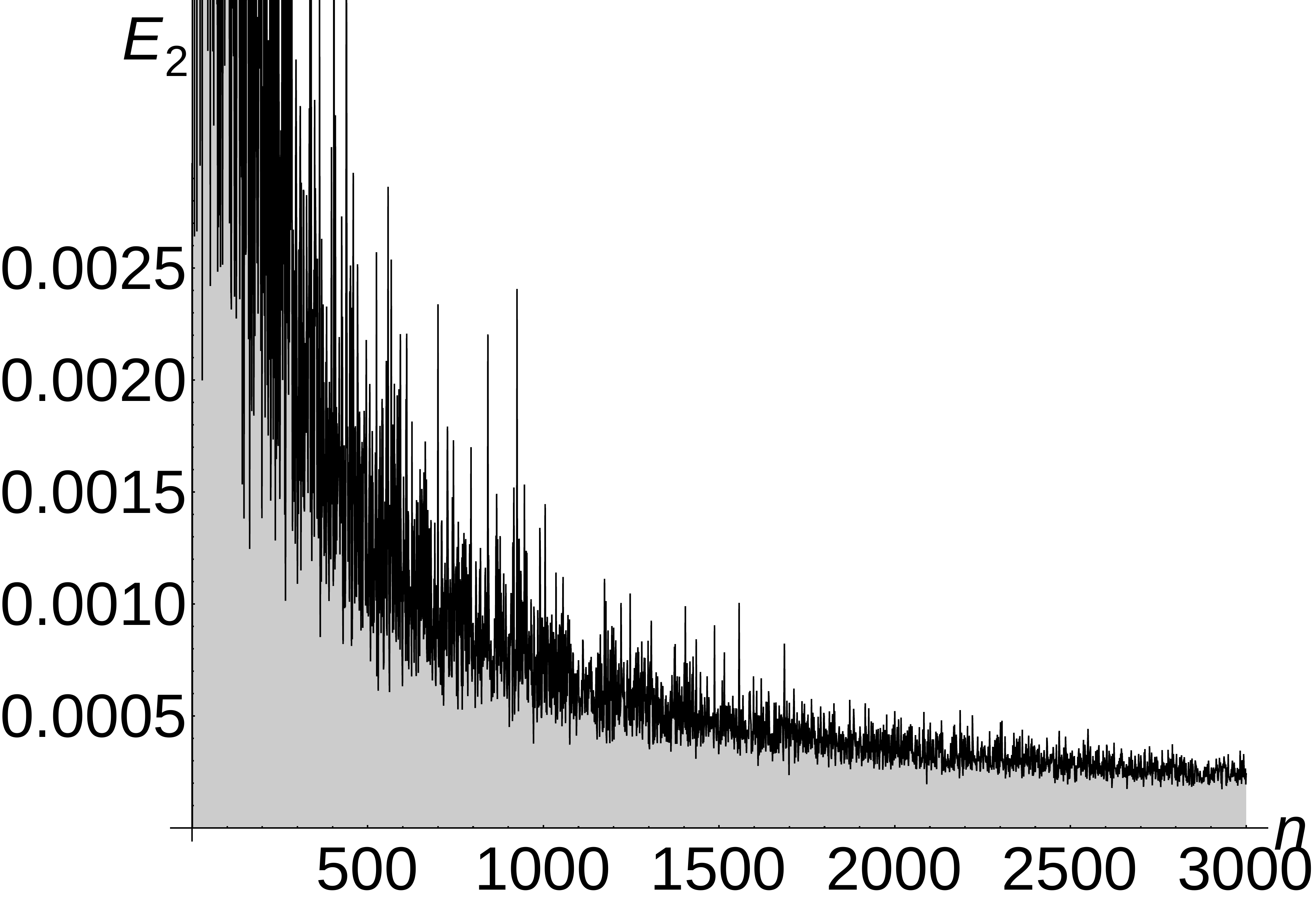}\\ ${E}_{2}(n)=\Theta\left(\frac{1}{n}\right)$\\
 \textbf{Case} $s<\frac{1}{n}$ \& $v>\frac{1}{n}.$
 \caption{The expected $2-$total movement ${E}_{2}(n)$ of \textbf{Algorithm \ref{alg_interference}} for $s=\frac{0.4}{n}$ and $v=\frac{1.6}{n}.$}\label{a:4}
 \end{minipage}
   \end{center}
 \end{figure}
 \begin{figure}[H]
 \label{fig.2b}
  \begin{center}
 \begin{minipage}[b]{0.43\textwidth} \centering \includegraphics[width=1.00\textwidth]{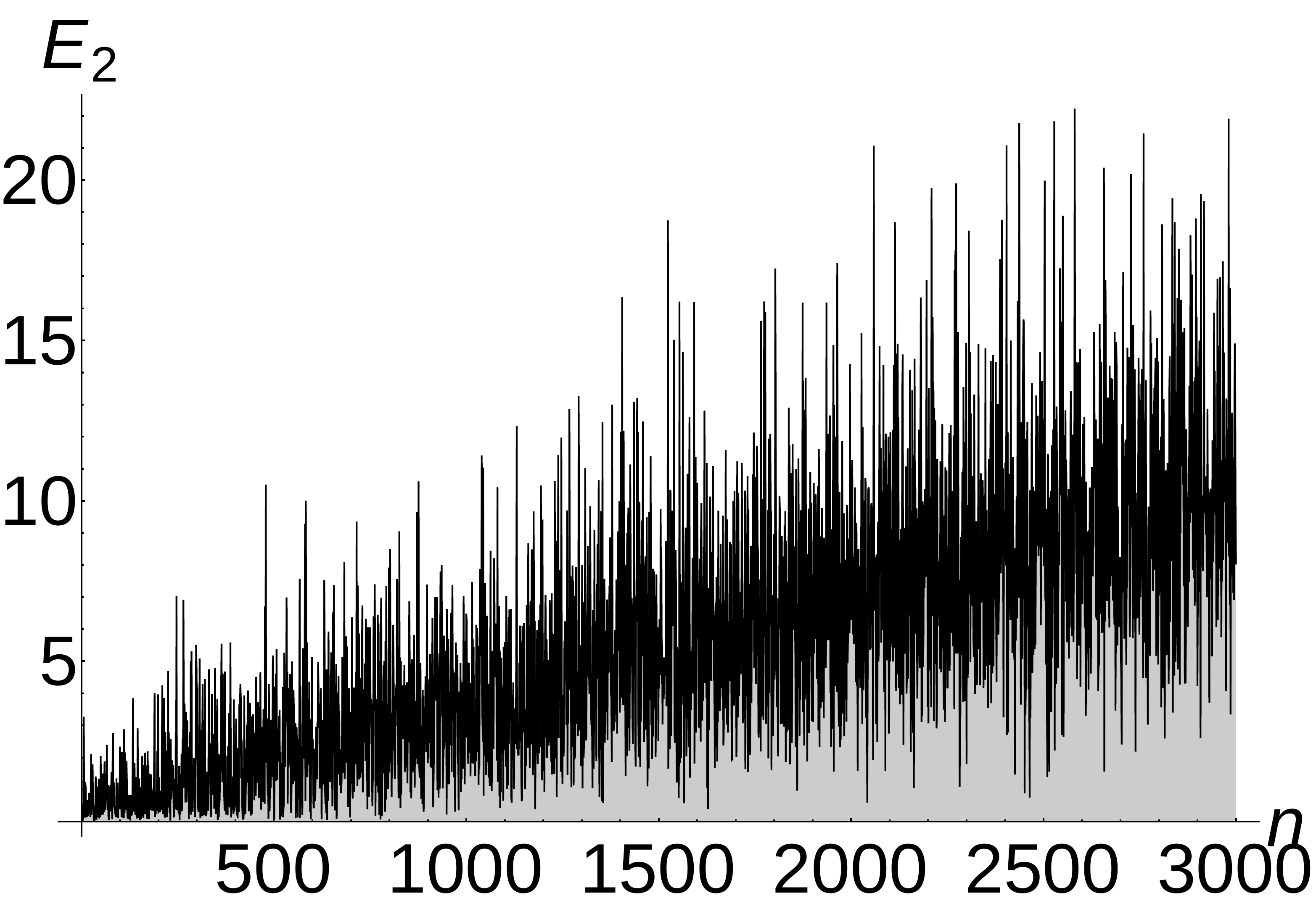}\\  ${E}_{2}(n)=\Theta(n)$\\
   \textbf{Case} $s>\frac{1}{n}$ \& $v>\frac{1}{n}.$
  \caption{The expected $2-$total movement ${E}_{2}(n)$ of \textbf{Algorithm \ref{alg_one}} for $s=\frac{1.1}{n}$}
 \end{minipage}
\end{center}
\end{figure}
\begin{figure}[H]
\label{fig.2c}
  \begin{center}
    \includegraphics[width=0.43\textwidth]{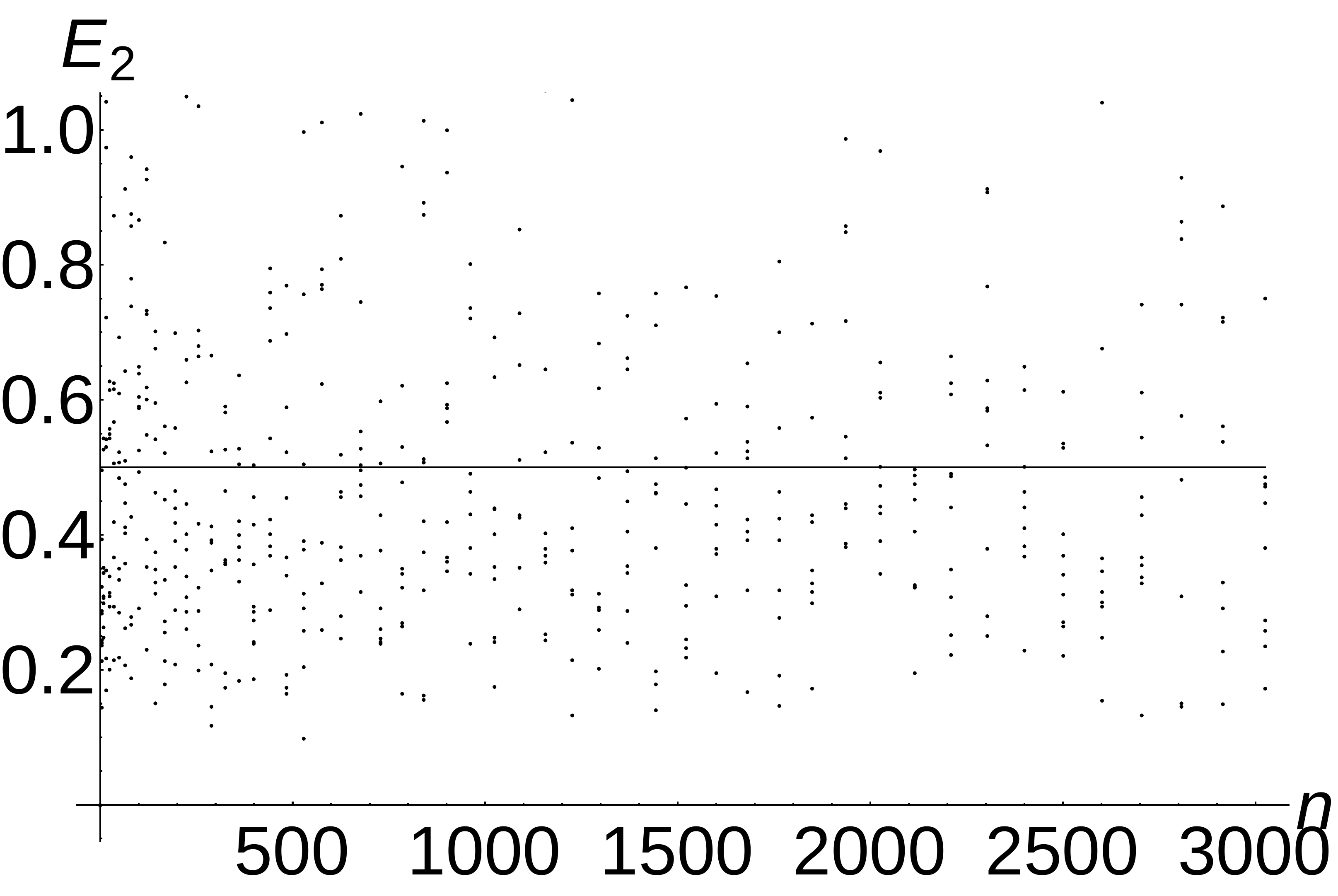}\\ ${E}_{2}(n)\sim \frac{1}{2}$\\
    \textbf{Case} $s=\frac{1}{n}$ \& $v=\frac{1}{n}.$
  \end{center}
  \caption{The expected $2-$total movement ${E}_{2}(n)$ of \textbf{Algorithm \ref{alg_one}} for $s=\frac{1}{n}$}
\end{figure}
Finally, it is worth to pointing out that the simulations confirm very well our theoretical tight, as well as upper bounds.
\section{Conclusion} In this paper we investigated the energy efficient displacement of random sensors
to avoid interference and to ensure good communication when $n$ sensors are initially placed in the hyperoctant $[0,\infty)^d$ according to $d$ identical and independent Poisson processes 
each with arrival rate $\lambda.$ 
We obtained tradeoffs between interference distances $s,$ $v$ and
the expected minimal $a-$total movement of $n$ random sensors. 
 It is discovered and explained {\textit{the threshold phenomena}} 
 around the interference-connectivity distance $s=v=\frac{1}{\lambda}$ 
 for the expected minimal $a-$total movement.
 
It would be interesting to study the energy efficient displacement of sensors
to ensure good communication and to avoid interference
for other more general random processes in one dimension, as well as in the higher dimensions.
 
There are also several interesting problems concerned with the ideal network when the consecutive sensors are not to close
while at the same time are not to far. These include performing efficient monitoring against illegal intruders
in one and in the higher dimensions. Another interesting problem is concerned with good communication when some sensors are unrelaible.
\section{Acknowledgment}
Research supported by grant nr 0401/0052/18 of FFPT, Wroc{\l}aw University of Science and Technology.

This work was also partially done during my scientific visit the School of Computer Science and the School of Mathematics and Statistics of Carleton Univeristy
from February 29 to March 11, 2016. I would like to thank Evangelos Kranakis and Gennady Shaikhet for fruitful discussions. 
\bibliographystyle{plain}
\bibliography{refs}
\newpage
\section*{Appendix}
\begin{proof} (Lemma \ref{lem:sum_tec})
Before providing the proof we recall some known facts about Stirling numbers, Eulerian numbers, as well as the finite difference operator that will be used in the proof.

Let $\left\langle\left\langle n\atop k\right\rangle\right\rangle$   be the Eulerian numbers of the second kind, which are
defined for all integer numbers such that $0\le k \le n.$ 
The following two identities for Eulerian numbers of the second kind are known (see Identities $(6.42),$  and $(6.44)$ in \cite{concrete_1994}):
\begin{equation}
\label{eq:euler1}
\sum_{l}\left\langle\left\langle m\atop l\right\rangle\right\rangle=\frac{(2m)!}{(m)!}\frac{1}{2^m},
\end{equation}
\begin{equation}
\label{eq:euler3}
{m\brack m-p}=\sum_{l}\left\langle\left\langle p\atop l\right\rangle\right\rangle\binom{m+l}{2p}.
\end{equation}

Let us recall the definition of the finite difference of a function f
$$\Delta f(x)=f(x+1)-f(x).$$  Then, high-order differences are defined by iteration
$$\Delta^af(x)=\Delta \Delta^{a-1}f(x).$$
It is easy to prove by induction the following formula (see also \cite[Identity 5.40]{concrete_1994})
\begin{equation}
\label{eq:a_diff}
\Delta^af(x)=\sum_{j}\binom{a}{j}(-1)^{a-j}f(x+j).
\end{equation}

We are now ready to proof Lemma \ref{lem:sum_tec}.

Choosing $f(x)={x\brack x-l_1}$ in (\ref{eq:a_diff}) we see that
\begin{align*}
\Delta^a {x\brack x-l_1}\Bigg|_{x=0}&=\sum_{j}\binom{a}{j}(-1)^{a-j}{x+j\brack x+j-l_1}\Bigg|_{x=0}\\
&=\sum_{j}\binom{a}{j}(-1)^{j}{j\brack j-l_1}.
\end{align*}
Applying equations (\ref{eq:euler1}), (\ref{eq:euler3}) and  the following identity 
$
\Delta^a \binom{x+l}{2l_1}\Big|_{x=0}= \begin{cases} 0 &\mbox{if } 2l_1 <a \\
1 & \mbox{if } 2l_1= a \end{cases},
$
we easily derive
\begin{align*}
\Delta^a {x\brack x-l_1}\Bigg|_{x=0}&=\sum_{l}\left\langle\left\langle l_1\atop l\right\rangle\right\rangle\Delta^a \binom{x+l}{2l_1}\Bigg|_{x=0}\\
&=\begin{cases} 0 &\mbox{if } 2l_1 <a \\
\sum_{l}\left\langle\left\langle \frac{a}{2}\atop l\right\rangle\right\rangle=\frac{a!}{\left(\frac{a}{2}\right)!2^{\frac{a}{2}}} & \mbox{if } 2l_1= a. \end{cases}
\end{align*}
This is enough to prove Lemma \ref{lem:sum_tec}. 
\end{proof}

\begin{proof}(Theorem \ref{thm:mainexactodd2})
Let us recall that, the asymptotic result of Theorem \ref{thm:mainexactodd2} for all exponents $a\ge 1$ follows from 
Theorem \ref{thm:mainexactodd} when $a$ is positive even natural
and the following probabilistic representation of absolute moments in terms of characteristic functions.

\begin{theorem}[cf. \cite{ushakov}, \cite{bahr}]  
\label{thm:maina}
Let $Y$ be a random variable with the distribution function $F(x)$ and the characteristic function $\varphi(t).$ Assume that
$\mathbf{E}\left[{|Y|^a}\right]<\infty,$ where $a\ge 1$ and $a$ is not an even integer. Let $\alpha_k=\mathbf{E}[Y^k],$ where $k$ is nonnegative integer. Then 
\begin{align*}
\mathbf{E}\left[|Y|^a\right]&=\frac{\Gamma(a+1)}{\pi}\cos\frac{(a+1)\pi}{2}\times\\
&\times\int_{-\infty}^{\infty}\frac{\Re\varphi(t)-\sum_{k=0}^{\left[\frac{a}{2}\right]}\frac{(-1)^k\alpha_{2k}t^{2k}}{(2k)!}}{|t|^{a+1}}dt,
\end{align*}
where $\left[\frac{a}{2}\right]$ is the greatest integer less than or equal to $\frac{a}{2}.$
\end{theorem}

As a first step, note that if $a\ge 1$ and a is an even integer the result of Theorem \ref{thm:mainexactodd2} follows from Theorem \ref{thm:mainexactodd}, as well as the identity $\Gamma(\frac{a}{2}+2)=\left(\frac{a}{2}+1\right)!.$

Therefore, we may assume that $a\ge 1$ and $a$ is not an even integer. Let $k$ be nonnegative integer.
Let $X_i$ be the arrival time of the $i-$th event in a Poisson process with arrival rate $n.$

We investigate the random variables $Y_i=\left(X_i-\frac{i-1}{n}\right)$
with its  characteristic function $\varphi_i(t),$  for $i = 2, \ldots , n.$ Observe that
\begin{equation}
\label{eq:charac}
\Re\varphi_i(t)-\sum_{k=0}^{\left[\frac{a}{2}\right]}\frac{(-1)^k\mathbf{E}[Y_i^{2k}]t^{2k}}{(2k)!}
=\sum_{k=\left[\frac{a}{2}\right]+1}^{\infty}\frac{(-1)^k\mathbf{E}[Y_i^{2k}]t^{2k}}{(2k)!}.
\end{equation}
Combining together Equation (\ref{eq:charac}) and Theorem \ref{thm:maina} for $Y:=Y_i$ we get
\begin{align}
\nonumber\mathbf{E}\left[|Y_i|^a\right]&=2\frac{\Gamma(a+1)}{\pi}\cos\frac{(a+1)\pi}{2}\times\\ 
&\times\label{eq:characer}\int_{0}^{\infty}\sum_{k=\left[\frac{a}{2}\right]+1}^{\infty}\frac{(-1)^k\mathbf{E}[Y_i^{2k}]t^{2k-a-1}}{(2k)!}dt.
\end{align}
Putting together Equation (\ref{eq:characer}) with Theorem \ref{thm:mainexactodd} for $a:=2k$ we derive
\begin{align*}
&\mathbf{E}\left[\sum_{i=2}^{n}|Y_i|^a\right]=2\frac{\Gamma(a+1)}{\pi}\cos\frac{(a+1)\pi}{2}\int_{0}^{\infty}\times\\ 
&\times\sum_{k=\left[\frac{a}{2}\right]+1}^{\infty}\frac{(-1)^kt^{2k-a-1}}{(2k)!}\left(\frac{(2k)!}{2^k(k+1)!}n^{1-k}+O\left(n^{-k}\right)\right)dt.
\end{align*}
Substitution $t=\sqrt{n}y$ in the last integral leads to
\begin{align}
\nonumber&\mathbf{E}\left[\sum_{i=2}^{n}|Y_i|^a\right]= 2\frac{\Gamma(a+1)}{\pi}\cos\frac{(a+1)\pi}{2}\times\\
&\nonumber\times\int_{0}^{\infty}\sum_{k=\left[\frac{a}{2}\right]+1}^{\infty}\frac{(-1)^k }{2^k(k+1)!}y^{2k-a-1}dy\left( n^{1-\frac{a}{2}}\right)\\
&\label{eq:intgrare01} +O\left( n^{-\frac{a}{2}}\right).
\end{align}
Let
\begin{align}
\nonumber C_a&=2\frac{\Gamma(a+1)}{\pi}\cos\frac{(a+1)\pi}{2}\times\\
\label{eq:intgrare02}&\times \int_{0}^{\infty}\sum_{k=\left[\frac{a}{2}\right]+1}^{\infty}\frac{(-1)^k}{2^k(k+1)!}y^{2k-a-1}dy. 
\end{align}
Using the identity
\begin{align}
\nonumber&\int_{0}^{\infty}\sum_{k=\left[\frac{a}{2}\right]+1}^{\infty}\frac{(-1)^k}{2^k(k+1)!}y^{2k-a-1}dy\\
\label{eq:intgrare03}&=\frac{-1}{2^{1+\frac{a}{2}}}\Gamma\left(-1-\frac{a}{2}\right)
\,\,\,\text{when}\,\,\,a\,\,\text{is not an even integer.}
\end{align}
we easily have
\begin{equation}
C_a=\frac{\Gamma(a+1)}{\pi}\cos\frac{(a+1)\pi}{2}
\label{eq:intgrare04} \frac{-1}{2^{\frac{a}{2}}}\Gamma\left(-1-\frac{a}{2}\right).
\end{equation}
\begin{Remark}
The following Mathematica code can be used to confirm the validity
of Identity (\ref{eq:intgrare03}).
\begin{verbatim}
Assuming[a>0 && 0<a-2IntegerPart[a/2]<2,
Integrate[Sum[((-1)^k)/(2^k*(k+1)!)
*y^(2k-a-1), {k, IntegerPart[a/2]+1, 
Infinity}],{y, 0, Infinity}]]
\end{verbatim}
\end{Remark}
From  Euler's reflection formula $\Gamma(1-z)\Gamma(z)=\frac{\pi}{\sin(\pi z)}$ (see \cite[Identity 5.5.3]{NIST})
for $z:=2+ \frac{a}{2}$, the identity $\frac{\cos\left(\frac{(a+1)\pi}{2}\right)}{\sin\left(\left(2+\frac{a}{2}\right)\pi\right)}=-1,$ as well as Equation (\ref{eq:intgrare04}) we derive
\begin{equation}
\label{eq:intgrare05}
C_a=\frac{\Gamma\left(a+1\right)}{2^{\frac{a}{2}}\Gamma\left(2+\frac{a}{2}\right)}.
\end{equation}
Together (\ref{eq:intgrare01}), (\ref{eq:intgrare02}) and (\ref{eq:intgrare05}) we conclude that
the expected $a-$total movement of algorithm $MV_1\left(n,\frac{1}{n}\right)$ is
$
\frac{\Gamma\left(a+1\right)}{2^{\frac{a}{2}}\Gamma\left(2+\frac{a}{2}\right)}n^{1-\frac{a}{2}}+O\left(n^{-\frac{a}{2}}\right)$
which completes the proof of Theorem \ref{thm:mainexactodd2}.
\end{proof}
\begin{proof} (Theorem \ref{thm:mainexactfractlower1a}) 
Assume that $a>1.$ Let $E^{(a)}_i$ be the expected
distance to the power a of $i-$th sensor for $i=1,2,\dots, n.$
Then we use discrete H\"older inequality with parameters $a$ and $\frac{a}{a-1}$
and get
\begin{eqnarray}
\sum_{i=1}^{n}E^{(1)}_i
&\le& \notag
\left(\sum_{i=1}^{n}\left(E^{(1)}_i\right)^{a}\right)^{\frac{1}{a}}
\left(\sum_{i=1}^{n}1\right)^{\frac{a-1}{a}}\\
&=&  \label{eq:holder2cc}
\left(\sum_{i=1}^{n}\left(E^{(1)}_i\right)^{a}\right)^{\frac{1}{a}}
n^{\frac{a-1}{a}}.
\end{eqnarray}
Next we use Jensen's inequality (see (\ref{eq:jensen})) for $f(x)=x^{a}$ and $\mathbf{E}[X]=E^{(1)}_i$
and get
\begin{equation}
\label{eq:jensen_b3a}
\left(E^{(1)}_i\right)^{a}\le E^{(a)}_i.
\end{equation}
Combining together (\ref{eq:holder2cc}), (\ref{eq:jensen_b3a}) and Theorem \ref{thm:mainexactfractlower}  we deduce that
\begin{align*}
\sum_{i=1}^{n}E^{(a)}_i&\ge \left(\sum_{i=1}^{n}E_{i}^{(1)}\right)^{a}
n^{-a+1}=\left(\Omega\left(\sqrt{n}\right)\right)^{a}
n^{-a+1}\\
&=
\Omega\left(n^{1-\frac{a}{2}}\right).
\end{align*}
This is enough to prove the lower bound and completes the proof of Theorem \ref{thm:mainexactfractlower1a}.
\end{proof}

\begin{proof} (Theorem \ref{thm:mainexactfract_epsilon} )
Let $D^{(a)}_i$ be the expected distance to the power $a$ between $X_i-X_1$ and the $i^{th}$ sensor position. 
Let $b$ be the even natural number such that $b-a>0.$
Then we
proceed as in the upper bound treatment from the proof of Theorem \ref{thm:mainexactfract} and get

\begin{equation}
 \label{eq:holder1eps}
\sum_{i=2}^{n}D^{(a)}_i\le
\left(\sum_{i=2}^{n}\left(D^{(a)}_i\right)^{\frac{b}{a}}\right)^{\frac{a}{b}}
(n-1)^{\frac{b-a}{b}}
\end{equation}
\begin{equation}
\label{eq:jensen_b2eps}
\left(D^{(a)}_i\right)^{\frac{b}{a}}\le D^{(b)}_i
\end{equation}
Combining  together (\ref{eq:holder1eps}), (\ref{eq:jensen_b2eps}) and Theorem \ref{thm:mainexactodd} we deduce that
$$
\sum_{i=2}^{n}D^{(a)}_i\le
\left(\Theta(n)\right)^{\frac{a}{b}}(n-1)^{\frac{b-a}{b}}=\Theta(n).
$$
This is sufficient to complete the proof of Theorem \ref{thm:mainexactfract_epsilon}.
\end{proof}
\begin{proof} (Theorem \ref{thm:mainexactfractlower_eps})

Fix $\tau\ge \epsilon >0$ independent on $n$ and $\lambda.$
Let $X_i$ be the arrival times of the $i-$th event in Poisson process with arrival rate. 
We assume that the algorithm moves the sensor $X_i$ to the position $b_i=Z+\frac{1+\Delta_i}{\lambda}(i-1)$ provided $\epsilon \le \Delta_i\le \tau$ for $i=1,2,\dots,n,$
where $Z$ is the nonnegative random variable with $\mathbf{E}[Z]<\infty.$

It is sufficient to show that
$$\sum_{i=1}^n\mathbf{E}\left[|X_i-b_i|\right]\in\frac{\Omega\left(n^2\right)}{\lambda}.$$ 

Applying Inequality (\ref{eq:first101}) 
for $Z:=X_i-Z-\Delta_i\frac{i-1}{\lambda},\,\,$  $q = \frac{i-1}{\lambda}$ and Equation (\ref{integral_2}) for $l=i$ we get
\begin{align*}
\sum_{i=1}^n&\mathbf{E}\left[\left|X_i-\left(Z+\frac{1+\Delta_i}{\lambda}(i-1)\right)\right|\right]\\
&\ge\sum_{i=1}^n\left|\frac{\mathbf{E}[\Delta_i]}{\lambda}i+\mathbf{E}[Z]-\frac{1+\mathbf{E}[\Delta_i]}{\lambda}\right|\\
&\ge\sum_{i=1}^n\left(\frac{\mathbf{E}[\Delta_i]}{\lambda}i+\mathbf{E}[Z]-\frac{1+\mathbf{E}[\Delta_i]}{\lambda}\right)\\
&\ge\sum_{i=1}^n\left(\frac{\epsilon}{\lambda}i+\mathbf{E}[Z]-\frac{1+\tau}{\lambda}\right)=\frac{\Theta\left(n^2\right)}{\lambda}
\end{align*}
This completes the proof of Theorem \ref{thm:mainexactfractlower_eps}. 
\end{proof}
\begin{proof} (Theorem \ref{thm:mainexactfractlower_eps1} )

Assume that $a>1.$ Let $E^{(a)}_i$ be the expected
distance to the power a of $i-$th sensor for $i=1,2,\dots, n.$
As in the proof of Theorem \ref{thm:mainexactfractlower1a} we get two inequalities:
\begin{equation}
\label{eq:holder2}
\sum_{i=1}^{n}E^{(1)}_i\le
\left(\sum_{i=1}^{n}\left(E^{(1)}_i\right)^{a}\right)^{\frac{1}{a}}
n^{\frac{a-1}{a}},
\end{equation}
\begin{equation}
\label{eq:jensen_b3}
\left(E^{(1)}_i\right)^{a}\le E^{(a)}_i.
\end{equation}
Combining together (\ref{eq:holder2}), (\ref{eq:jensen_b3}) and Theorem \ref{thm:mainexactfractlower_eps}  we deduce that
$$
\sum_{i=1}^{n}E^{(a)}_i\ge \left(\sum_{i=1}^{n}E_{i}^{(1)}\right)^{a}
n^{-a+1}=\left(\frac{\Omega\left(n^2\right)}{\lambda}\right)^{a}
n^{-a+1}=
\frac{\Omega\left(n^{1+a}\right)}{\lambda^a}.
$$
This is enough to prove the lower bound and completes the proof of Theorem \ref{thm:mainexactfractlower_eps1}.
\end{proof}

\begin{proof} (Theorem \ref{thm:interfere} )
Let $s=\frac{1-\epsilon}{\lambda}, v=\frac{1+\tau}{\lambda},$ where $\epsilon$ and $\tau$ are arbitrary small constants independent on $n$ and $\lambda.$

Firstly, we consider the following scenario. 
\begin{itemize}
\item Algorithm \ref{alg_interference} leaves the sensors $X_1, X_2,\dots,X_{i}$ at the same positions. 
\item Algorithm \ref{alg_interference} moves the sensors $X_{i+1}, X_{i+2}, \dots X_{i+q}$ at the new locations.
\item Algorithm \ref{alg_interference} leaves the sensors $X_{i+q+1}, X_{i+q+2},\dots,X_{n}$ at the same positions.
\end{itemize}
Let $q=q_1+q_2+\dots+q_k$ for some $q_1,q_2,\dots,q_k\in\mathbb{N_+}.$
Define $$S_j=\begin{cases}i\,\,\,\text{if}\,\,\, j=0\\
i+q_1+q_2+\dots +q_j\,\,\,\text{if}\,\,\,j\in\{1,2,\dots, k-1\}.
\end{cases}$$
Consider an integer configuration $(q_1,q_2,\dots,q_k)$ as specified above.
\begin{itemize}
\item For each $j\in\{0,1,\dots, k-1\}$  Algorithm \ref{alg_interference} moves the sensors
$X_{S_j+1}, X_{S_j+2},\dots, X_{S_j+q_{j+1}}$ in the one chosen direction
left to right or right to left. 
\item For each $j\in\{0,1,\dots, k-2\}$  the movement direction of sensors $X_{S_j+1}, X_{S_j+2},\dots, X_{S_j+q_{j+1}}$ is opposite to the movement direction
of sensors $X_{S_{j+1}+1}, X_{S_{j+1}+2},\dots, X_{S_{j+1}+q_{j+2}}.$
\end{itemize}
Let $T(q_1,q_2,\dots, q_k)$ be the movement to the power $a$ of the sensors $X_{i+1}, X_{i+2}, \dots X_{i+q}$
and let $T(q_k)$ be the displacement to the power $a$ of the sensors $X_{S_{k-1}+1}, X_{S_{k-1}+2},\dots, X_{S_{k-1}+q_{k}}.$
There are two cases to consider.

\textit{Case 1. The sensor $X_{S_{k-1}}$ moves left to right to the new position $Y$ and
the sensors $X_{S_{k-1}+1}, X_{S_{k-1}+2},\dots, X_{S_{k-1}+q_{k}}$ move right to left.}

Observe that
$$T(q_k)=\sum_{l=1}^{q_k}\left(\left|X_{S_{k-1}+l}-\left(Y+vl\right)\right|^{+}\right)^a.$$
Since $X_{S_{k-1}}<Y$ we upper bound the displacement $T(q_k)$ as follows
$$T(q_k)\le\sum_{l=1}^{q_k}\left(\left|X_{S_{k-1}+l}-\left(X_{S_{k-1}}+vl\right)\right|^{+}\right)^a.$$
Using Identity $X_{j+l}-X_{j}=X_l$ for $j:=S_{k-1}$ (see (\ref{eq:probadens})) we get
\begin{equation}
\label{eq:tab}
T(q_k)\le \sum_{l=1}^{q_k} (|X_{l}-vl|^{+})^a.
\end{equation}

\textit{Case 2. The sensor $X_{S_{k-1}}$ moves right to left to the new position $Z$ and
the sensors $X_{S_{k-1}+1}, X_{S_{k-1}+2},\dots, X_{S_{k-1}+q_{k}}$ move left to right.}

Observe that
$$T(q_k)=\sum_{l=1}^{q_k}\left(\left|Z+sl-X_{S_{k-1}+l}\right|^{+}\right)^a.$$
Since $Z<X_{S_{k-1}}$ we upper bound the displacement $T(q_k)$ as follows
$$T(q_k)=\sum_{l=1}^{q_k}\left(\left|X_{S_{k-1}}+sl-X_{S_{k-1}+l}\right|^{+}\right)^a.$$
Using Identity $X_{j+l}-X_{j}=X_l$ for $j:=S_{k-1}$ (see (\ref{eq:probadens})) we get
\begin{equation}
\label{eq:tabb}
T(q_k)\le \sum_{l=1}^{q_k} (|sl-X_{l}|^{+})^a.
\end{equation}

Combining together inequalities (\ref{eq:tab}) and (\ref{eq:tabb}) we have
\begin{align*}
T(q_1,q_2,\dots, q_k)\le & T(q_1,q_2,\dots, q_{k-1})+ \sum_{l=1}^{q_k} (|X_{l}-vl|^{+})^a\\
+& \sum_{l=1}^{q_k} (|sl-X_{l}|^{+})^a.
\end{align*}
Hence, by induction we get
\begin{align*}
T(q_1,&q_2,\dots, q_k)\\\ 
&\le \sum_{j=1}^{k}\sum_{l=1}^{q_j} (|X_{l}-vl|^{+})^a+ \sum_{j=1}^k\sum_{l=1}^{q_j} (|sl-X_{l}|^{+})^a.
\end{align*}
Next we make an important observation that extends our estimation to general scenario of Algorithm \ref{alg_interference}. 
Let $T_a$ be the $a-$total displacement of Algorithm \ref{alg_interference}.

Observe that the displacements $\left(|sl-X_{l}|^{+}\right)^a$ and $\left(|X_{l}-vl|^{+}\right)^a$ can appear in the Algorithm \ref{alg_interference} at most $\frac{n}{l}$ times. Therefore
$$
T_a\le\sum_{l=1}^{n}\frac{n}{l}(|sl-X_{l}|^{+})^a+\sum_{l=1}^{n}\frac{n}{l}(|X_{l}-vl|^{+})^a
$$
Passing to the expectations we have
$$
\mathbf{E}\left[T_a\right]\le\sum_{l=1}^{n}\frac{n}{l} \mathbf{E}\left[(|sl-X_{l}|^{+})^a\right]+\sum_{l=1}^{n}\frac{n}{l} \mathbf{E}\left[(|X_{l}-vl|^{+})^a\right]
$$
Finally, using  Lemma \ref{lemma_d} for $s=\frac{1-\epsilon}{\lambda}$ and  $v=\frac{1+\tau}{\lambda}$
we conclude that
$\mathbf{E}\left[T_a\right]=\frac{O\left(n\right)}{\lambda^a}.$ This is enough to prove Theorem \ref{thm:interfere}.
\end{proof}
\begin{proof} (Lemma \ref{lemma_d})

Fix $a\ge 1.$ There are two cases to consider

\textit{Case 1: Inequality (\ref{eq:lem_d1})}

Assume that $X_l$ obeys Gamma distribution with parameters $l\in\mathbf{N}\setminus\{0\}$ and $\lambda>0.$
Let $s=\frac{1-\epsilon}{\lambda},$ where $\epsilon>0$ is some constant independent on $n$ and $\lambda.$

The following technical inequality for Gamma distribution is known.
\begin{align}
\nonumber\label{eq:betalena}\mathbf{E}&\left[\left((sl-X_{l})^{+}\right)^a\right]\\
&\le\left(\frac{1-\epsilon}{\lambda}\right)^a\left(l^a(le+2)\left((1-\epsilon)e^{\epsilon}\right)^l+\frac{l^a}{\epsilon}\frac{\left((1-\epsilon)e^{\epsilon}\right)^l}{e^l}\right)
\end{align}
(see the estimation for $k:=l$ and $\delta=\epsilon$ after Inequality (35) in \cite{pervasiveKAPELKO}).

Using elementary inequality $(1-\epsilon)e^{\epsilon}<1$ when $\epsilon\in(0,1)$ we deduce that
\begin{align}
\nonumber\sum_{l=1}^{n}&\frac{1}{l}\left(\frac{1-\epsilon}{\lambda}\right)^a\left(l^a(le+2)\left((1-\epsilon)e^{\epsilon}\right)^l+\frac{l^a}{\epsilon}\frac{\left((1-\epsilon)e^{\epsilon}\right)^l}{e^l}\right)\\
\label{eq:astra01}&=\frac{O(1)}{\lambda^a}.
\end{align}
Combining together (\ref{eq:betalena}), (\ref{eq:astra01}) we get
$$
\sum_{l=1}^{n}\frac{n}{l} \mathbf{E}\left[(|sl-X_{l}|^{+})^a\right]=\frac{O\left(n\right)}{\lambda^a}.
$$
This is enough to prove Inequality (\ref{eq:lem_d1}).

\textit{Case 2: Inequality (\ref{eq:lem_d2})}

Assume that $X_l$ obeys Gamma distribution with parameters $l\in\mathbf{N}\setminus\{0\}$ and $\lambda>0.$
Let $v=\frac{1+\tau}{\lambda},$ where $\tau>0$ is some constant independent on $n$ and $\lambda.$

The following technical inequality for Gamma distribution is also known.
\begin{align}
\nonumber\label{eq:betalenb}\mathbf{E}&\left[\left((X_{l}-v l)^{+}\right)^{ a}\right]\\
&\frac{1}{\lambda^{a}}\left(f_{a,\tau}(l)\right)^{\frac{a}{\lceil a\rceil}}\left(\left(\frac{1+\tau}{e^{\tau}}\right)^{\frac{a}{\lceil a\rceil}}\right)^l,
\end{align}
where
\begin{align*}
&f_{a,\tau}(l)\\
&=(l+\lceil a\rceil)^{\lceil a\rceil}\left((l+1)+(\lceil a\rceil-1)(1+\epsilon)^{\lceil a\rceil-1}l^{\lceil a\rceil-1}\right)
\end{align*}
(see the estimation for $k:=l$ and $\epsilon:=\tau$ after Inequality (27) in \cite{pervasiveKAPELKO}).

Applying elementary inequality $\frac{1+\tau}{e^{\tau}}<1,$ when $\tau>0$ we get
$\left(\frac{1+\tau}{e^{\tau}}\right)^{\frac{a}{\lceil a\rceil}}<1.$ Since $\left(f_{a,\tau}(l)\right)^{\frac{a}{\lceil a\rceil}}$ is bounded by
some polynomial fo variable $l$ and degree $l^{\lceil a\rceil+1}$ we deduce that
\begin{equation}
\sum_{l=1}^{n}\frac{1}{\lambda^{a}}\frac{\left(f_{a,\tau}(l)\right)^{\frac{a}{\lceil a\rceil}}}{l}
\left(\left(\frac{1+\tau}{e^{\tau}}\right)^{\frac{a}{\lceil a\rceil}}\right)^l
\label{eq:astra02}=\frac{O(1)}{\lambda^a}.
\end{equation}
Combining together (\ref{eq:betalenb}), (\ref{eq:astra02}) we get
$$
\sum_{l=1}^{n}\frac{n}{l} \mathbf{E}\left[(|X_{l}-vl|^{+})^a\right]=\frac{O\left(n\right)}{\lambda^a}.
$$

This is enough to prove desired Inequality (\ref{eq:lem_d2}) and it completes the proof of Lemma \ref{lemma_d}.
\end{proof}
\begin{proof} (Lemma \ref{lem:cruciana})
First of all, we recall the following elementary inequality.

Fix $a\ge 1.$ Let $x,y \ge 0.$ Then
\begin{equation}
\label{eq:duszal}
(x+y)^a\le 2^{a-1}(x^a+y^a).
\end{equation}
Notice that, Inequality (\ref{eq:duszal}) is the consequence of the fact that $f(x)=x^a$ is convex over $\mathbf{R_+}$ for $a\ge 1.$

Applying $(d-1)$ times Inequality (\ref{eq:duszal}) for the sequence $M_1, M_2, \dots M_d$ and passing to the expectations we easily derive 
$$\mathbf{E}[M^a]\le C_{a,d}\sum_{i=1}^{d}\mathbf{E}[M_i^a].$$
This is enough to prove Lemma \ref{lem:cruciana}.
\end{proof}

\begin{proof} (Theorem \ref{thm:mainexacthigh_eps})

Fix $\tau\ge \epsilon>0$ independent on $n.$
By Theorem \ref{thm:mainexactfract_epsilon} applied to
$n:=n^{1/d}$ and 
for $n^{(d-1)/d}$ columns and $n^{(d-1)/d}$ rows, as well as Lemma \ref{lem:cruciana} 
we have that the expected $a-$total movement of algorithm $MV_d\left(n,\frac{1+\epsilon}{\lambda}\right)$ is 
$$2C_{a,d}n^{{(d-1)}/d}\frac{O\left(\left(n^{1/d}\right)^{1+a}\right)}{\lambda^a}=\frac{O\left(n^{1+\frac{a}{d}}\right)}{\lambda^a}.$$
This completes the prove of the upper bound.

By Theorem \ref{thm:mainexactfractlower_eps} and Theorem \ref{thm:mainexactfractlower_eps1} applied to
$n:=n^{1/d}$ and 
for $n^{(d-1)/d}$ columns 
we have that the following lower bound
$$n^{{(d-1)}/d}\frac{\Omega\left(\left(n^{1/d}\right)^{1+a}\right)}{\lambda^a}=\frac{\Omega\left(n^{1+\frac{a}{d}}\right)}{\lambda^a}.$$
This is enough to prove Theorem \ref{thm:mainexacthigh_eps}.
\end{proof}
\begin{proof} (Theorem \ref{thm:mainexacthigh_eps_minus} )

Fix $\epsilon, \tau>0$ independent on $n.$ By Theorem \ref{thm:interfere} applied to
$n:=n^{1/d}$ and 
for $n^{(d-1)/d}$ columns and $n^{(d-1)/d}$ rows , as well as Lemma \ref{lem:cruciana} 
we have the following upper bound
$$2C_{a,d}n^{{(d-1)}/d}\frac{O\left(n^{1/d}\right)}{\lambda^a}=\frac{O\left(n\right)}{\lambda^a},$$
which proves the theorem.
\end{proof}
\end{document}